\documentclass{vldb}

\usepackage{graphicx}
\usepackage{balance}  

\vldbTitle{CLX: Towards verifiable PBE data transformation}
\vldbAuthors{Zhongjun Jin, Michael Cafarella, H. V. Jagadish, Sean Kandel, Michael Minar, Joseph M. Hellerstein}
\vldbDOI{https://doi.org/TBD}
\vldbVolume{12}
\vldbNumber{xxx}
\vldbYear{2019}

\usepackage{color}
\usepackage{caption}
\usepackage{subcaption} 
\usepackage{amsmath}
\usepackage{fancyvrb}

\usepackage{graphicx}
\usepackage{adjustbox}

\usepackage{textcomp}

\usepackage{tabularx}
\usepackage{booktabs}
\usepackage{multirow}

\usepackage[ruled,vlined,linesnumbered]{algorithm2e}

\usepackage{diagbox}

\usepackage{pgfplots}

\usepackage[inline,shortlabels]{enumitem}

\usepackage[utf8]{inputenc}
\usepackage{hyperref}

\usepackage{siunitx}

\usepackage{listings} 

\usepackage[utf8]{inputenc}

\usepackage{color}

 \usepackage{url}

\setlist[itemize]{noitemsep, topsep=0pt}
\setlist[enumerate]{noitemsep, topsep=0pt}

\usetikzlibrary{patterns}

\definecolor{dkgreen}{rgb}{0,0.6,0}
\definecolor{gray}{rgb}{0.5,0.5,0.5}
\definecolor{mauve}{rgb}{0.58,0,0.82}
\newcommand{\system}{{\sc CLX}}
\newcommand{\tde}{{\sc TDE}}
\newcommand{\lang}{{\sc UniFi}}
\newcommand{\wrangler}{{\sc Wrangler}}
\newcommand{\baseline}{{\sc FlashFill}}
\newcommand{\blinkfill}{{\sc BlinkFill}}
\newcommand{\foofah}{{\sc Foofah}}
\newcommand{\learnpads}{{\sc LearnPADS}}
\newcommand{\datamaran}{{\sc Datamaran}}

\newcommand{\trifacta}{{\sc Trifacta}}

\newcommand{\trifactawrangler}{{\sc TrifactaWrangler}}

\newcommand{\regexreplace}{{\sc RegexReplace}}

\newcommand{\Naively}{Na\"{i}vely}
\newcommand{\naive}{na\"{i}ve}

\newcommand*{\affmark}[1][*]{\textsuperscript{#1}}

\newcommand{\minisection}[1]{\vspace{0.25cm} \noindent {\bf #1} ---}

\newcommand{\informalsection}[1]{\vspace{0.15cm} \noindent {\em #1}. }

\newtheorem{theorem}{Theorem}[section]

\newtheorem{example}{Example}

\newtheorem{defn}{Definition}[section]

\newcommand{\eg}{{e.g.}}
\newcommand{\Eg}{{E.g.}}
\newcommand{\ie}{{i.e.}}

\newcommand{\ea}{{et al.}}

\newcommand{\etal}{{et al.}}

\newcounter{choice}
\renewcommand\thechoice{\Alph{choice}}
\newcommand\choicelabel{\thechoice.}

\newenvironment{choices}%
  {\list{\choicelabel}%
     {\usecounter{choice}\def\makelabel##1{\hss\llap{##1}}%
       \settowidth{\leftmargin}{W.\hskip\labelsep\hskip 2.5em}%
       \def\choice{%
         \item
       } 
       \labelwidth\leftmargin\advance\labelwidth-\labelsep
       \topsep=0pt
       \partopsep=0pt
     }%
  }%
  {\endlist}

  {%
    \setcounter{choice}{0}%
    \def\choice{%
      \refstepcounter{choice}%
      \ifnum\value{choice}>1\relax
        \penalty -50\hskip 1em plus 1em\relax
      \fi
      \choicelabel
      \nobreak\enskip
    }
    \ifvmode\else\enskip\fi
    \ignorespaces
  }%
  {}

\begin{document}
\title{CLX: Towards verifiable PBE data transformation}

\numberofauthors{1} 
\author{
\alignauthor
Zhongjun Jin\affmark[1]\quad\ Michael Cafarella\affmark[1]\quad\ H. V. Jagadish\affmark[1]\quad\ Sean Kandel\affmark[2] \\
Michael Minar\affmark[2]\quad\ Joseph M. Hellerstein\affmark[2,3]\\
      \affaddr{\affmark[1]University of Michigan, Ann Arbor \quad\ \affmark[2]Trifacta Inc. \quad\ \affmark[3]UC Berkeley}\\
      \email{\{markjin,michjc,jag\}@umich.edu \quad\ \{skandel,mminar\}@trifacta.com\quad\ hellerstein@berkeley.edu}
}

\maketitle

\begin{abstract}
Effective data analytics on data collected from the real world usually begins with a notoriously expensive pre-processing step of data transformation and wrangling. Programming By Example (PBE) systems have been proposed to automatically infer transformations using simple examples that users provide as hints. However, an important usability issue---\textbf{verification}---limits the effective use of such PBE data transformation systems, since the verification process is often effort-consuming and unreliable.

We propose a data transformation paradigm design \system\ (pronounced ``clicks'') with a focus on facilitating verification for end users in a PBE-like data transformation.  \system\ performs pattern clustering in both input and output data, which allows the user to verify at the pattern level, rather than the data instance level, without having to write any regular expressions, thereby significantly reducing user verification effort. Thereafter, \system\ automatically generates transformation programs as regular-expression replace operations that are easy for average users to verify.

We experimentally compared the \system\ prototype with both \baseline, a state-of-the-art PBE data transformation tool, and \trifacta, an influential system supporting interactive data transformation. The results show improvements over the state of the art tools in saving user verification effort, without loss of efficiency or expressive power. In a user study on data sets of various sizes, when the data size grew by a factor of 30, the user verification time required by the \system\ prototype grew by $1.3\times$ whereas that required by \baseline\ grew by $11.4\times$.  In another user study assessing the users' understanding of the transformation logic --- a key ingredient in effective verification --- \system\ users achieved a success rate about twice that of \baseline\ users.
\end{abstract}


\section{Introduction}\label{sec:intro}
Data transformation, or data wrangling, is a critical pre-processing step essential to effective data analytics on real-world data and is widely known to be human-intensive as it usually requires professionals to write ad-hoc scripts that are difficult to understand and maintain. A \textit{human-in-the-loop} Programming By Example (PBE) approach has been shown to reduce the burden for the end user: in projects such as \baseline~\cite{Gulwani2011}, \blinkfill~\cite{Singh2016}, and \foofah~\cite{Jin2017}, the system synthesizes data transformation programs using simple examples the user provides. 

\minisection{Problems}
Most of existing research in PBE data transformation tools has focused on the ``system'' part --- improving the efficiency and expressivity of the program synthesis techniques. Although these systems have demonstrated some success in efficiently generating high-quality data transformation programs for real-world data sets, \textbf{verification}, as an indispensable interaction procedure in PBE, remains a major bottleneck within existing PBE data transformation system designs. The high labor cost, may deter the user from confidently using these tools. 

Any reasonable user who needs to perform data transformation should certainly care about the ``correctness'' of the inferred transformation logic. In fact, a user will typically go through rounds of ``verify-and-specify" cycles when using a PBE system. 
In each interaction, a user has to verify the correctness of the current inferred transformation logic by validating the transformed data instance by instance until she identifies a data instance mistakenly transformed; then she has to provide a new example for correction. \textbf{Given a potentially large and varied input data set, such a verification process is like ``finding a needle in a haystack'' which can be extremely time-consuming and tedious.}

A \naive\ way to simplify the cumbersome verification process is to add explanations to the transformed data so that the user does not have to read them in their raw form. For example, if we can somehow know the desired data pattern, we can write a checking function to automatically check if the post-transformed data satisfies the desired pattern, and highlight data entries that are not correctly transformed.

However, a {\em data explanation} procedure alone can not solve the entire verification issue; the undisclosed transformation logic remains untrustworthy to the end user. Users can at best verify that existing data are converted into the right form, but \textbf{the logic is not guaranteed to be correct and may function unexpectedly on new input} (see Section~\ref{sec:motivating} for an example). Without good insight into the transformation logic, PBE system users cannot tell if the inferred transformation logic is correct, or when there are errors in the logic, they may not be able to debug it. \textbf{If the user of a traditional PBE system lacks good understanding of the synthesize program's logic, she can only verify it by spending large amounts of time testing the synthesized program on ever-larger datasets}.

\Naively, previous PBE systems can support {\em program explanation} by  presenting the inferred programs to end users. However, these data transformation systems usually design their own Domain Specific Languages (DSLs), which are usually sophisticated. The steep learning curve makes it unrealistic for most users to quickly understand the actual logic behind the inferred programs. Thus, besides more explainable data, a desirable PBE system should be able to present the transformation logic in a way that most people are already familiar with.

\minisection{Insight}
Regular expressions (regexp) have been known to most programmers of various expertise and regexp replace operations have been commonly applied in data transformations. The influential data transformation system, \wrangler\ (later as \trifacta), proposes simplified natural-language-like regular expressions which can be understood and used even by non-technical data analysts. This makes regexp replace operations a good choice for an {\em explainable transformation language}. The challenge then is how to automatically synthesize regexp replace operations as the desired transformation logic in a PBE system. 

A regexp replace operation takes in two parameters: an {\em input pattern} and a {\em replacement function}. Suppose an input data set is given, and the desired data pattern can be known, the challenge is to determine a suitable input pattern and the replacement function to convert all input data into the desired pattern. Moreover, if the input data set is heterogeneous with many formats, we need to find out an unknown set of such input-pattern-and-replace-function pairs.

Pattern profiling can be used to discover clusters of data patterns within a data set that are are useful to generate regular replace operations. Moreover, it can also serve as a data explanation approach helping the user quickly understand the pre- and post-transformation data which reduces the verification challenge users face in PBE systems. 

\minisection{Proposed Solution}
In this project, we propose a new data transformation paradigm, \system, to address the two specific problems within our claimed verification issue as a whole. The \system\ paradigm has three components: two algorithmic components---\textit{clustering} and \textit{transformation}---with an intervening component of \textit{labeling}. In this paper, we present an instantiation of the \system\ paradigm. We present 
\begin{enumerate*}[(1)]
\item an efficient pattern clustering algorithm that groups data with similar structures into small clusters, 
\item a DSL for data transformation, that can be interpreted as a set of regular expression replace operations,  
\item a program synthesis algorithm to infer desirable transformation logic in the proposed DSL.
\end{enumerate*}

Through the above means, we are able to greatly ameliorate the usability issue in verification within PBE data transformation systems. 
Our experimental results show improvements over the state of the art in saving user verification effort, along with increasing users' comprehension of the inferred transformations. Increasing comprehension is highly relevant to reducing the verification effort. In one user study on a large data set, when the data size grew by a factor of 30, the \system\ prototype cost $1.3 \times$ more verification time whereas \baseline\ cost $11.4 \times$ more verification time. In a separate user study accessing the users' understanding of the transformation logic, \system\ users achieved a success rate about twice that of \baseline\ users. Other experiments also suggest that the expressive power of the \system\ prototype and its efficiency on small data are comparable to those of \baseline. 


\minisection{Organization}
After motivating our problem with an
example in Section~\ref{sec:motivating}, we discuss the following contributions:
\begin{itemize}
\itemsep0em 
 \item We define the data transformation problem and present the PBE-like CLX framework solving this problem. (Section~\ref{sec:overview})
  \item We present a data pattern profiling algorithm to hierarchically cluster the raw data based on patterns.
(Section~\ref{sec:profiling})
\item We present a new DSL for data pattern transformation in the \system\ paradigm. (Section~\ref{sec:standardizationProg})
  \item We develop algorithms synthesizing data transformation programs, which can transform any given input pattern to the desired standard pattern. (Section~\ref{sec:algorithm})
  \item We experimentally evaluate the \system\ prototype and other baseline systems through user studies and simulations. (Section~\ref{sec:evaluation})
\end{itemize}
We explore the related work in Section~\ref{sec:related} and finish with a
discussion of future work in Section~\ref{sec:conclusion}.


\section{Motivating Example}\label{sec:motivating}

\begin{figure}[t]
\centering

\begin{minipage}[b]{0.48\linewidth}
\begin{minipage}[b]{\textwidth}
\centering
\small
\begin{adjustbox}{max width=.45\textwidth}
\begin{tabular}{|l|}
\hline
(734) 645-8397 \\ \hline
(734)586-7252 \\ \hline
734-422-8073 \\ \hline
734.236.3466 \\ \hline
...
\end{tabular}
\end{adjustbox}
\caption{Phone numbers with diverse formats}
\label{tab:rawData}
\end{minipage}

\begin{minipage}[b]{\textwidth}
\centering
\small
 \begin{adjustbox}{max width=\textwidth}
\begin{tabular}[t]{|l|}
\hline
\verb|\({digit}3\)\ {digit}3\-{digit}4| \\
\textbf{(734) 645-8397} ... \textit{(10000 rows)}\\ \hline
\end{tabular}
\end{adjustbox}
\caption{Patterns after transformation}\label{tab:posttransform}
\end{minipage}
\end{minipage}\hfill
\begin{minipage}[b]{0.5\linewidth}
\centering
\small
 \begin{adjustbox}{max width=\textwidth}
\begin{tabular}[t]{|l|}
\hline
\verb|\({digit}3\){digit}3\-{digit}4| \\
\textbf{(734)586-7252} ... \textit{(2572 rows)}\\ \hline
\verb|{digit}3\-{digit}3\-{digit}4| \\
\textbf{734-422-8073} ... \textit{(3749 rows)}\\ \hline
\verb|\({digit}3\)\ {digit}3\-{digit}4| \\
\textbf{(734) 645-8397} ... \textit{(1436 rows)}\\ \hline
\verb|{digit}3\.{digit}3\.{digit}4| \\
\textbf{734.236.3466} ... \textit{(631 rows)}\\ \hline
...
\end{tabular}
\end{adjustbox}
\vspace{0.3cm}
\caption{Pattern clusters of raw data}
\label{fig:unifacta}

\end{minipage}

\begin{minipage}{\linewidth}
\centering

\lstset{
basicstyle=\small\tt,
keywordstyle=\color{dkgreen},
numberstyle=\small}
\begin{lstlisting}
Replace '/^\(({digit}{3})\)({digit}{3})\-({digit}{4})$/' in column1 with '($1) $2-$3'
Replace '/^({digit}{3})\-({digit}{3})\-({digit}{4})$/' in column1 with '($1) $2-$3'
...
\end{lstlisting}
\caption{Suggested data transformation operations}
\label{fig:unifactaScript}
\end{minipage}
\end{figure}
Bob is a technical support employee at the customer service department. He wanted to have a set of 10,000 phone numbers in various formats (as in Figure~\ref{tab:rawData}) in a unified format of ``(xxx) xxx-xxxx''. Given the volume and the heterogeneity of the data, neither manually fixing them or hard-coding a transformation script was convenient for Bob. He decided to see if there was an automated solution to this problem.

Bob found that Excel 2013 had a new feature named \baseline\ that could transform data patterns. He loaded the data set into Excel and performed \baseline\ on them.

\begin{example}
Initially, Bob thought using \baseline\ would be straightforward: he would simply need to provide an example of the transformed form of each ill-formatted data entry in the input and copy the exact value of each data entry already in the correct format. However, in practice, it turned out not to be so easy. First, Bob needed to carefully check each phone number entry deciding whether it is ill-formatted or not. After obtaining a new input-output example pair, \baseline\ would update the transformation results for the entire input data, and Bob had to carefully examine again if any of the transformation results were incorrect. This was tedious given the large volume of heterogeneous data \textbf{(verification at string level is challenging)}. After rounds of repairing and verifying, Bob was finally sure that \baseline\ successfully transformed all existing phone numbers in the data set, and he thought the transformation inferred by \baseline\ was impeccable. Yet, when he used it to transform another data set, a phone number ``+1 724-285-5210'' was mistakenly transformed as ``(1) 724-285'', which suggested that the transformation logic may fail anytime \textbf{(unexplainable transformation logic functions unexpectedly)}. Customer phone numbers were critical information for Bob's company and it was important not to damage them during the transformation. With little insight from \baseline\ regarding the transformation program generated, Bob was not sure if the transformation was reliable and had to do more testing \textbf{(lack of understanding increases verification effort)}.
\end{example}

Bob heard about \system\ and decided to give it a try.
\begin{example}
He loaded his data into \system\ and it immediately presented a list of distinct string patterns for phone numbers in the input data (Figure~\ref{fig:unifacta}), which helped Bob quickly tell which part of the data were ill-formatted. After Bob selected the desired pattern, \system\ immediately transformed all the data and showed a new list of string patterns as Figure~\ref{tab:posttransform}. \textbf{So far, verifying the transformation result was straightforward.} The inferred program is presented as a set of \textsf{Replace} operations on raw patterns in Figure~\ref{fig:unifacta}, each with a picture visualizing the transformation effect. Bob was not a regular expressions guru, but these operations seemed simple to understand and verify. Like many users in our User Study (Section~\ref{subsec:explanationEval}), \textbf{Bob had a deeper understanding of the inferred transformation logic with \system\ than with \baseline,  and hence, he knew well when and how the program may fail, which saved him from the effort of more blind testing.}
\end{example}


% ...


\section{Overview}\label{sec:overview}
\begin{table}[t]
\centering
\fontsize{8}{8}\selectfont
\setlength{\tabcolsep}{2pt}
\begin{tabularx}{\linewidth}{lX}
\toprule
Notation&Description   \\ \midrule
$\mathcal{S} = \{\mathit{s}_1, \mathit{s}_2, \dots\}$ & A set of ad hoc strings $s_1, s_2, \dots$ to be transformed. \\
$\mathcal{P} = \{\mathit{p}_1, \mathit{p}_2, \dots\}$& A set of string patterns derived from $\mathcal{S}$. \\

$\mathit{p}_i = \{\mathit{t}_1, \mathit{t}_2, \dots\}$ &Pattern made from a sequence of tokens $t_i$ \\
 
$\mathcal{T}$ & The desired target pattern that all strings in $\mathcal{S}$ needed to be transformed into. \\
$\mathcal{L} = \{(\mathit{p}_1, \mathit{f}_1), (\mathit{p}_2, \mathit{f}_2), \dots\}$ & Program synthesized in \system\ transforming data the patterns of $\mathcal{P}$ into $\mathcal{T}$.\\
$\mathcal{E}$ & The expression $\mathcal{E}$ in $\mathcal{L}$, which is a concatenation of \textsf{Extract} and/or \textsf{ConstStr} operations. It is a transformation plan for a source pattern. We also refer to it as an \textit{Atomic Transformation Plan} in the paper.\\
$\mathcal{Q}(\mathfrak{\widetilde{t}},p)$ & Frequency of token $\mathfrak{\widetilde{t}}$ in pattern $p$\\
$\mathcal{G}$ & Potential expressions represented in Directed Acyclic Graph.\\
\bottomrule
\end{tabularx}
\caption{Frequently used notations}\label{tab:notation}
\end{table}

\subsection{Patterns and Data Transformation Problem}
A data pattern, or string pattern, is a ``high-level'' description of the attribute value’s string. A natural way to describe a pattern could be a regular expression over the characters that constitute the string. In data transformation, we find that groups of contiguous characters are often transformed together as a group. Further, these groups of characters are meaningful in themselves.  For example, in a date string ``11/02/2017'', it is useful to cluster ``2017'' into a single group, because these four digits are likely to be manipulated together. We call such meaningful groups of characters as {\em tokens}.

Table~\ref{tab:tokenDescription} presents all {\em token classes} we currently support in our instantiation of \system, including their class names, regular expressions, and notation. In addition, we also support tokens of constant values (\eg, ``,'', ``:''). In the rest of the paper, we represent and handle these tokens of constant values differently from the 5 token classes defined in Table~\ref{tab:tokenDescription}. For convenience of presentation, we denote such tokens with constant values as \textbf{\textit{literal tokens}} and tokens of 5 token classes defined in Table~\ref{tab:tokenDescription} as \textbf{\textit{base tokens}}. 

A pattern is written as a sequence of tokens, each followed by a quantifier indicating the number of occurrences of the preceding token. A quantifier is either a single natural number or ``$+$'', indicating that the token appears at least once. In the rest of the paper, to be succinct, a token will be denoted as ``$\langle\mathfrak{\widetilde{t}}\rangle \mathfrak{q}$'' if $\mathfrak{q}$ is a number (\eg, $\langle D\rangle$3) or ``$\langle\mathfrak{\widetilde{t}}\rangle+$'' otherwise (\eg, $\langle D\rangle$+). If $\mathfrak{\widetilde{t}}$ is a literal token, it will be surrounded by a single quotation mark, like `:'.  When a pattern is shown to the end user, it is presented as a {\em natural-language-like regular expression} proposed by \wrangler~\cite{Kandel2011} (see regexps in  Fig~\ref{fig:unifactaScript}).

\begin{table}[t]
\centering
\fontsize{8}{8}\selectfont
\begin{tabularx}{\linewidth}{llll}
\toprule
Token Class & Regular Expression & Example & Notation\\ 
\midrule
digit & \texttt{[0-9]} & ``12'' & $\langle D \rangle$ \\
lower & \texttt{[a-z]} & ``car'' & $\langle L\rangle$ \\
upper & \texttt{[A-Z]} & ``IBM'' & $\langle U\rangle$ \\
alpha & \texttt{[a-zA-Z]} & ``Excel'' & $\langle A\rangle$ \\
alpha-numeric &\texttt{[a-zA-Z0-9\_-]} & ``Excel2013'' & $\langle AN \rangle$ \\
\bottomrule
\end{tabularx}
\caption{Token classes and their descriptions}
\label{tab:tokenDescription}
\end{table}

With the above definition of data patterns, we hereby formally define the problem we tackle using the \system\ framework---data transformation. Data transformation or wrangling is a broad concept. Our focus in this paper is to apply the \system\ paradigm to transform a data set of heterogeneous patterns into a desired pattern. A formal definition of the problem is as follows:
\begin{defn}[Data (Pattern) Transformation]\label{defn:normalization}
Given a set of strings $\mathcal{S}=\{\mathit{s}_1, \dots, \mathit{s}_n \}$, generate a program $\mathcal{L}$ that transforms each string in $\mathcal{S}$ to an equivalent string matching the user-specified desired target pattern $\mathcal{T}$.
\end{defn}

$\mathcal{L} = \{(\mathit{p}_1, \mathit{f}_1), (\mathit{p}_2, \mathit{f}_2), \dots\}$ is the program we synthesize in the transforming phase of \system. It is represented as a set regexp replace operations, \textsf{Replace}$(\mathit{p}, \mathit{f})$\footnote{$\mathit{p}$ is the regular expression, and $\mathit{f}$ is the replacement string indicating the operation on the string matching the pattern $\mathit{p}$.}, that many people are familiar with (\eg, Fig~\ref{fig:unifactaScript}). 

With above definitions of patterns and data transformations, we present the \system\ framework for data transformation.

\subsection{\system\ Data Transformation Paradigm}
We propose a data transformation paradigm called Cluster-Label-Transform (\system, pronounced ``clicks''). Figure~\ref{fig:interaction} visualizes the interaction model in this framework.  

\minisection{Clustering}
The clustering component groups the raw input data into clusters based on their data patterns/formats. Compared to raw strings, data patterns is a more abstract representation. The number of patterns is fewer than raw strings, and hence, it can make the user understand the data and verify the transformation more quickly. Patterns discovered during clustering is also useful information for the downstream program synthesis algorithm to determine the number of regexp replace operations, as well as the desirable input patterns and transformation functions. 

\minisection{Labeling}
Labeling is to specify the desired data pattern that every data instance is supposed to be transformed into.
Presumably, labeling can be achieved by having the user choose among the set of patterns we derive in the clustering process assuming some of the raw data already exist in the desired format. If no input data matches the target pattern, the user could alternatively choose to manually specify the target data form.

\minisection{Transforming}
After the desired data pattern is labeled, the system automatically synthesizes data transformation logic that transforms all undesired data into the desired form and also proactively helps the user understand the transformation logic. 

\begin{figure}[t]
\centering
\includegraphics[width=.9\linewidth]{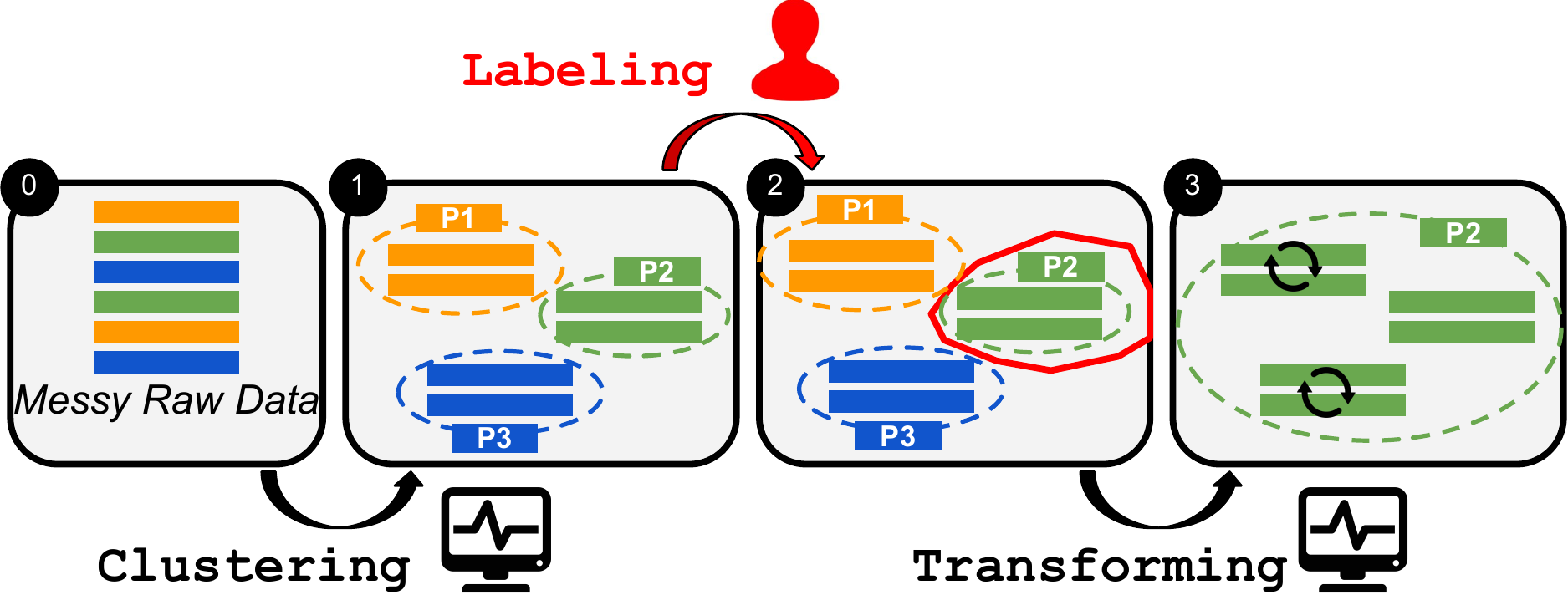}
\caption{``CLX'' Model: Cluster--Label--Transform}
\label{fig:interaction}
\end{figure}

In this paper, we present an instantiation of the \system\ paradigm for \textit{data pattern transformation}. Details about the clustering component and the transformation component are discussed in Section~\ref{sec:profiling} and \ref{sec:algorithm}. In Section~\ref{sec:standardizationProg}, we show the domain-specific-language (DSL) we use to represent the program $\mathcal{L}$ as the outcome of program synthesis, which can be then presented as the regexp replace operations. 
The paradigm has been designed to allow new algorithms and DSLs for transformation problems other than data pattern transformation; we will pursue other instantiations in future work.

\section{Clustering data on patterns}\label{sec:profiling}
In \system, we first cluster data into meaningful groups based on their structure and obtain the pattern information, which helps the user quickly understand the data. To minimize user effort, this clustering process should ideally not require user intervention. 

\learnpads~\cite{Fisher2008} is an influential project that also targets string pattern discovery. However, \learnpads\ is orthogonal to our effort in that their goal is mainly to find a comprehensive and unified description for the entire data set whereas we seek to partition the data into clusters, each cluster with a single data pattern.
Also, the PADS language~\cite{Fisher2005} itself is known to be hard for a non-expert to read~\cite{zhu2012learnpads}. Our interest is to derive simple patterns that are comprehensible. Besides the explainability, efficiency is another important aspect of the clustering algorithm we must consider, because input data can be huge and the real-time clustering must be interactive.


To that end, we propose an automated means to hierarchically cluster data based on data patterns given a set of strings. The data is clustered through a two-phase profiling: 
\begin{enumerate*}[(1)]
\item tokenization: tokenize the given set of strings of ad hoc data and cluster based on these initial patterns, 
\item agglomerative refinement: recursively merge pattern clusters to formulate a \textbf{\textit{pattern cluster hierarchy}} that allows the end user to view/understand the pattern structure information in a simpler and more systematic way, and also helps \system\ generate a simple transformation program.
\end{enumerate*}

\subsection{Initial Clustering 
Through Tokenization}
Tokenization is a common process in string processing when string data needs to be manipulated in chunks larger than single characters. A simple parser can  do the job.

Below are the rules we follow in the tokenization phase.
\begin{itemize}
\itemsep0em 
\item Non-alphanumeric characters carry important hints about the string structure. Each such character is identified as an individual literal token. 
\item We always choose the most precise base type to describe a token. For example, a token with string content ``cat'' can be categorized as ``lower'', ``alphabet'' or ``alphanumeric'' tokens. We choose ``lower'' as the token type for this token.
\item The quantifiers are always natural numbers.
\end{itemize}
Here is an example of the token description of a string data record discovered in tokenization phase.
\begin{example}\label{ex:token}
Suppose the string ``Bob123@gmail.com'' is to be tokenized. The result of tokenization becomes [$\langle U\rangle$, $\langle L\rangle 2$, $\langle D\rangle 3$, `@', $\langle L \rangle 5$, `.', $\langle L \rangle 3$].
\end{example}

After tokenization, each string corresponds to a data pattern composed of tokens. We create the initial set of pattern clusters by clustering the strings sharing the same patterns. Each cluster uses its pattern as a label which will later be used for refinement, transformation, and user understanding. 

\minisection{Find Constant Tokens} Some of the tokens in the discovered patterns have constant values.  Discovering such constant values and representing them using the actual values rather than base tokens helps improve the quality of the program synthesized. For example, if most entities in a faculty name list contain ``Dr.'', it is better to represent a pattern as [`Dr.',`$\backslash$ ', `$\langle U \rangle$', `$\langle L \rangle$+'] than [`$\langle U \rangle$', `$\langle L \rangle$', `.', `$\backslash$ ', `$\langle U \rangle$', `$\langle L \rangle$+']. Similar to \cite{Fisher2008}, we find tokens with constant values using the statistics over tokenized strings in the data set.


\begin{algorithm}[t]
\SetKwData{strategy}{$\tilde{g}$}
\SetKwData{pattern}{$\mathcal{P}$}
\SetKwData{parent}{$p_{parent}$}
\SetKwData{cand}{$\mathcal{P}_{raw}$}
\SetKwData{count}{$\mathcal{C}_{raw}$}
\SetKwData{return}{$\mathcal{P}_{final}$}

 \KwData{Pattern set \pattern, generalization strategy \strategy}
\KwResult{Set of more generic patterns \return}
\return,\cand $\leftarrow \emptyset$\;
\count $\leftarrow \{\}$\;
\For{ $p_i \in$ \pattern}{\label{algo:refine:1:start}
\parent\ $\leftarrow\ getParent(p_i, \strategy)$\;
add \parent\ to \cand\;\label{algo:refine:1:mid}
\count[\parent] = \count[\parent] + 1  \;\label{algo:refine:1:end}
}
\For{ \parent\ $\in$ \cand\ ranked by \count\ from high to low}{\label{algo:refine:2:start}
 $\parent.child  \leftarrow \{ p_j | \forall p_j \in \pattern, p_{j}.isChild(\parent)\}$\; 
 add \parent\ to \return\;\label{algo:refine:2:end}
 remove $\parent.child$ from \pattern\;\label{algo:refine:5:end}
}
\textbf{Return} \return\;\label{algo:refine:3}

 \caption{Refine Pattern Representations}
 \label{algo:refine}
 
\end{algorithm}
\subsection{Agglomerative Pattern Cluster Refinement} 
In the initial clustering step, we distinguish different patterns by token classes, token positions, and quantifiers, the actual number of pattern clusters discovered in the ad hoc data in tokenization phase could be huge. User comprehension is inversely related to the number of patterns. It is not very helpful to present too many very specific pattern clusters all at once to the user. Plus, it can be unacceptably expensive to develop data pattern transformation programs separately for each pattern. 

To mitigate the problem, we build {\em pattern cluster hierarchy}, \ie, a hierarchical pattern cluster representation with the leaf nodes being the patterns discovered through tokenization, and every internal node being a {\em parent pattern}. With this hierarchical pattern description, the user can understand the pattern information at a high level without being overwhelmed by many details, and the system can generate simpler programs. Plus, we do not lose any pattern discovered previously.

From bottom-up, we recursively cluster the patterns at each level to obtain {\em parent patterns}, \ie, more generic patterns, formulating the new layer in the hierarchy. To build a new layer, Algorithm~\ref{algo:refine} takes in different generalization strategy $\tilde{g}$ and the child pattern set $\mathcal{P}$ from the last layer. 
Line~\ref{algo:refine:1:start}-\ref{algo:refine:1:mid} clusters the current set of pattern clusters to get parent pattern clusters using the generalization strategy $\widetilde{g}$. The generated set of parent patterns may be identical to others or might have overlapping expressive power. Keeping all these parent patterns in the same layer of the cluster hierarchy is unnecessary and increases the complexity of the hierarchy generated. Therefore, we only keep a small subset of the parent patterns initially discovered and make sure they together can cover any child pattern in $\mathcal{P}$.
To do so, we use a counter $\mathcal{C}_{raw}$ counting the frequencies of the obtained parent patterns (line~\ref{algo:refine:1:end}). Then, we iteratively add the  parent pattern that covers the most patterns in $\mathcal{P}$ into the set of more generic patterns to be returned (line~\ref{algo:refine:2:start}-\ref{algo:refine:5:end}). The returned set covers all patterns in $\mathcal{P}$ (line~\ref{algo:refine:3}). Overall, the complexity is $\mathcal{O}(n \log n)$, where $n$ is the number of patterns in P, and hence, the algorithm itself can quickly converge.

In this paper, we perform three rounds of refinement to construct the new layer in the hierarchy, each with a particular generalization strategy:
\begin{enumerate}
\itemsep0em 
\item natural number quantifier to `+'
\item $\langle L \rangle$, $\langle U \rangle$ tokens to $\langle A \rangle$
\item $\langle A \rangle$, $\langle N \rangle$, '-', `\_' tokens to $\langle AN \rangle$
\end{enumerate}

\begin{figure}[t]
\centering
\includegraphics[width=.8\linewidth]{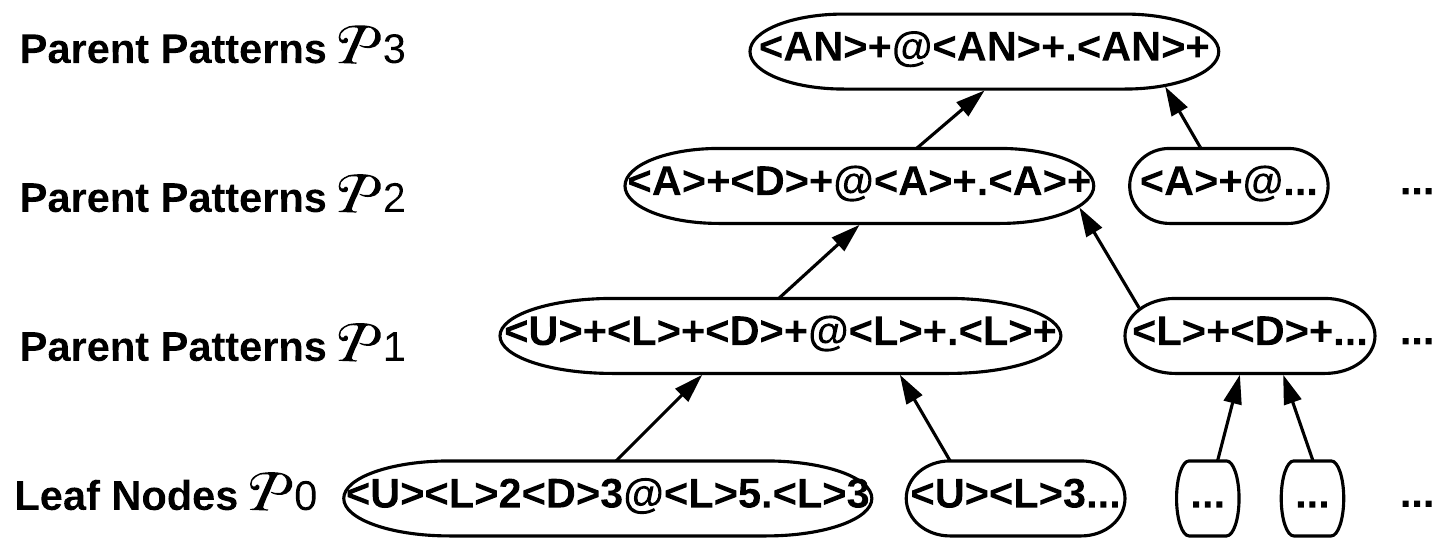}
\caption{Hierarchical clusters of data patterns}
\label{fig:pattern-tree}
\end{figure}

\begin{example}\label{ex:token2}
Given the pattern we obtained in Example~\ref{ex:token}, we successively apply Algorithm~\ref{algo:refine} with Strategy 1, 2 and 3 to generalize parent patterns $\mathcal{P}_1$, $\mathcal{P}_2$ and $\mathcal{P}_3$ and construct the pattern cluster hierarchy as in Figure~\ref{fig:pattern-tree}.
\end{example}

\subsection{Limitations}
The pattern hierarchy constructed can succinctly profile the pattern information for many data. However, the technique itself may be weak in two situations. First, as the scope of this paper is limited to addressing the syntactic transformation problem (Section~\ref{sec:standardizationProg}), the pattern discovery process we propose only considers syntactic features, but no semantic features. This may introduce the issue of ``misclustering''. For example, a date of format ``MM/DD/YYYY'' and a date of format ``DD/MM/YYYY'' may be grouped into the same cluster of ``$\langle N \rangle 2 / \langle N \rangle 2 / \langle N \rangle 4$'', and hence, transforming from the former format into the latter format is impossible in our case. Addressing this problem requires the support for semantic information discovery and transformation, which will be in our future work. Another possible weakness of ``fail to cluster'' is also mainly affected by the semantics issue: we may fail to cluster semantically-same but very messy data. \Eg, we may not cluster the local-part (everything before `@') of a very weird email address ``Mike\textquotesingle John.Smith@gmail.com'' (token $\langle AN \rangle$ cannot capture `\textquotesingle' or `.'). Yet, this issue can be easily resolved by adding additional regexp-based token classes (\eg, emails). Adding more token classes is beyond the interest of our work.







\section{Data Pattern Transformation Program }\label{sec:standardizationProg}

As motivated in Section~\ref{sec:intro} and Section~\ref{sec:overview}, our proposed data transformation framework is to synthesize a set of regexp replace operations that people are familiar with as the desired transformation logic. However, representing the logic as regexp strings will make the program synthesis difficult. Instead, to simplify the program synthesis, we propose a new language, \lang, as a representation of the transformation logic internal to \system.  The grammar of \lang\ is shown in Figure~\ref{fig:program}. We then discuss how to explain a inferred \lang\ program as regexp replace operations.

The top-level of any \lang\ program is a \textsf{Switch} statement that conditionally maps strings to a transformation. \textsf{Match} checks whether a string $s$ is an \textit{exact match} of a certain pattern  $p$ we discover previously. Once a string matches this pattern, it will be processed by an {\em atomic transformation plan} (expression $\mathcal{E}$ in \lang) defined below.

\begin{defn}[Atomic Transformation Plan]\label{def:atp}
Atomic transformation plan is a sequence of parameterized string operators that converts a given source pattern into the target pattern.
\end{defn}

The available string operators include \textsf{ConstStr} and \textsf{Extract}. \textsf{ConstStr($\widetilde{s}$)} denotes a constant string $\widetilde{s}$. \textsf{Extract}($\mathfrak{\widetilde{t}}_i$,$\mathfrak{\widetilde{t}}_j$) extracts from the $i^{\text{th}}$ token to the $j^{\text{th}}$ token in a pattern. In the rest of the paper, we express an \textsf{Extract} operation as \textsf{Extract}($i$,$j$), or \textsf{Extract}($i$) if $i=j$. A token $\mathfrak{t}$ is represented as $(\mathfrak{\widetilde{t}}, \mathfrak{r}, \mathfrak{q}, i)$: $\mathfrak{\widetilde{t}}$ is the token class in Table~\ref{tab:tokenDescription}; $\mathfrak{r}$ represents the corresponding regular expression of this token; $\mathfrak{q}$ is the quantifier of the token expression; $i$ denotes the index (one-based) of this token in the source pattern. 

\begin{figure}[t]
\centering
\begin{equation*} \label{eq1}
\begin{split}
\text{Program } \mathcal{L} :=&\ \textsf{Switch}((b_1, \mathcal{E}_1), \dots, (b_n, \mathcal{E}_n)) \\
\text{Predicate } b :=&\ \textsf{Match}(s, p)\\
\text{Expression } \mathcal{E} :=&\ \textsf{Concat}(f_1, \dots, f_n)\\
\text{String Expression } f :=&\ \textsf{ConstStr}(\widetilde{s})\ |\ \textsf{Extract}(\mathfrak{\widetilde{t}}_i,\mathfrak{\widetilde{t}}_j)\\
 \text{Token Expression } \mathfrak{t}_{i} :=&\ (\mathfrak{\widetilde{t}}, \mathfrak{r}, \mathfrak{q}, i)\\
\end{split}
\end{equation*}
\caption{\lang\ Language Definition}
\label{fig:program}
\end{figure}


As with \baseline~\cite{Gulwani2011} and \blinkfill~\cite{Singh2016}, we only focus on syntactic transformation, where strings are manipulated as a sequence of characters and no external knowledge is accessible, in this instantiation design. {\em Semantic transformation} (\eg, converting ``March'' to ``03'') is a subject for future work. Further--again like \blinkfill~\cite{Singh2016}--our proposed data pattern transformation language \lang\ does not support loops. Without the support for loops, \lang\ may not be able to describe transformations on an unknown number of occurrences of a given pattern structure.

We use the following two examples used by \baseline~\cite{Gulwani2011} and \blinkfill~\cite{Singh2016} to briefly demonstrate the expressive power of \lang, and the more detailed expressive power of \lang\ would be examined in the experiments in Section~\ref{subsec:expressiveness}. For simplicity, \textsf{Match}($s,p$) is shortened as \textsf{Match}($p$) as the input string $s$ is fixed for a given task.



\begin{example}\label{ex:missingRightBracket}
This problem is modified from test case ``\textbf{Example 3}'' in \blinkfill. The goal is to transform all messy values in the medical billing
codes into the correct form ``[CPT-XXXX]'' as in Table~\ref{tab:medical}.

\begin{table}[h]
\centering
\fontsize{8}{8}\selectfont
\begin{tabular}[t]{|l|l|}
\hline
\textbf{Raw data} & \textbf{Transformed data} \\ \hline
CPT-00350 & [CPT-00350]  \\ \hline
[CPT-00340 & [CPT-00340] \\ \hline
[CPT-11536] & [CPT-11536] \\ \hline
CPT115 & [CPT-115]\\ \hline
\end{tabular}
\caption{Normalizing messy medical billing codes}
\label{tab:medical}
\end{table}
The \lang\ program for this standardization task is
\begin{Verbatim}[fontsize=\small, fontfamily=courier]
Switch((Match("\[<U>+\-<D>+"),
     (Concat(Extract(1,4),ConstStr(']')))),
   (Match("<U>+\-<D>+"),
     (Concat(ConstStr('['),Extract(1,3), 
       ConstStr(']'))))
   (Match("<U>+<D>+"),
     (Concat(ConstStr('['),Extract(1),
     ConstStr('-'),Extract(2),ConstStr(']')))))
\end{Verbatim}
\end{example}

\begin{example}
This problem is borrowed from ``\textbf{Example 9}'' in \baseline. The goal is to transform all names into a unified format as in Table~\ref{tab:employee}.

\begin{table}[h]
\centering
\fontsize{8}{8}\selectfont
\begin{tabular}[t]{|l|l|}
\hline
\textbf{Raw data} & \textbf{Transformed data} \\ \hline
Dr. Eran Yahav & Yahav, E.  \\ \hline
Fisher, K. & Fisher, K. \\ \hline
Bill Gates, Sr. & Gates, B. \\ \hline
Oege de Moor & Moor, O.\\ \hline
\end{tabular}
\caption{Normalizing messy employee names}
\label{tab:employee}
\end{table}
A \lang\ program for this task is 
\begin{Verbatim}[fontsize=\small, fontfamily=courier]
Switch((Match("<U><L>+\.\ <U><L>+\ <U><L>+"),
    Concat(Extract(8,9),ConstStr(','), 
        ConstStr(' '),Extract(5))),
   (Match("<U><L>+\ <U><L>+\,\ <U><L>+\."),
    Concat(Extract(4,5),ConstStr(','), 
        ConstStr(' '),Extract(1))),
   (Match("<U><L>+\ <U>+\ <U><L>+"),
    Concat(Extract(6,7),ConstStr(','),
        ConstStr(' '),Extract(1))))
\end{Verbatim}

\end{example}

\minisection{Program Explanation}~\label{sub:explanation}
Given a \lang\ program $\textsf{L}$, we want to present it as a set of regexp replace operations, \textsf{Replace}, parameterized by {\em natural-language-like regexps} used by Wrangler~\cite{Kandel2011} (\eg,  Figure~\ref{fig:unifactaScript}), which are straightforward to even non-expert users. Each component of $(b, \mathcal{E})$, within the $\textsf{Switch}$ statement of $\textsf{L}$, will be explained as a \textsf{Replace} operation. The replacement string $f$ in the \textsf{Replace} operation is created from $p$ and the transformation plan $\mathcal{E}$ for the condition $b$. In $f$, a \textsf{ConstStr}$(\tilde{s})$ operation will remain as $\tilde{s}$, whereas a \textsf{Extract}$(\tilde{t_i}, \tilde{t_j})$ operation will be interpreted as $\$\tilde{t_i}\dots\$\tilde{t_j}$. The pattern $p$ in the predicate $b = \textsf{Match}(s,p)$ in \lang\ naturally becomes the regular expression $p$ in \textsf{Replace} with each tokens to be extracted surrounded by a pair of parentheses indicating that it can be extracted. Note that if multiple consecutive tokens are extracted in $p$, we merge them as one component to be extracted in $p$ and change the $f$ accordingly for convenience of presentation. Figure~\ref{fig:unifactaScript} is an example of the transformation logic finally shown to the user.

In fact, these \textsf{Replace} operations can be further explained using visualization techniques. For example, we could add a \textit{Preview Table} (\eg, Figure~\ref{fig:preview}) to visualize the transformation effect in our prototype in a sample of the input data. The user study in Section~\ref{subsec:explanationEval} demonstrates that our effort of outputting an explainable transformation program helps the user understand the transformation logic generated by the system. 
\section{Program Synthesis}
\label{sec:algorithm}

We now discuss how to find the desired transformation logic as a \lang\ program using the pattern cluster hierarchy obtained. Algorithm~\ref{algo:progSyn} shows our synthesis framework. 

Given a pattern hierarchy, we do not need to create an \textit{atomic transformation plan} (Definition~\ref{def:atp}) for every pattern cluster in the hierarchy.  We traverse the pattern cluster hierarchy top-down to find valid \textit{candidate source patterns} (line~\ref{algo:progSyn:1:validsource}, see Section~\ref{subsubsec:cand}). Once a source candidate is identified, we discover all {\em token matches} between this source pattern in $Q_{solved}$ and the target pattern (line~\ref{algo:progSyn:1}, see Section~\ref{subsubsec:tokenAlign}). With the generated token match information, we synthesize the data pattern normalization program including an atomic transformation plan for every source pattern (line~\ref{algo:progSyn:2}, see Section~\ref{subsec:rank}). 

\subsection{Identify Source Candidates}\label{subsubsec:cand}
\begin{figure*}[!t]
\centering
\begin{minipage}{0.20\textwidth}
\centering
\vspace{0.17cm}
\includegraphics[width=.95\linewidth]{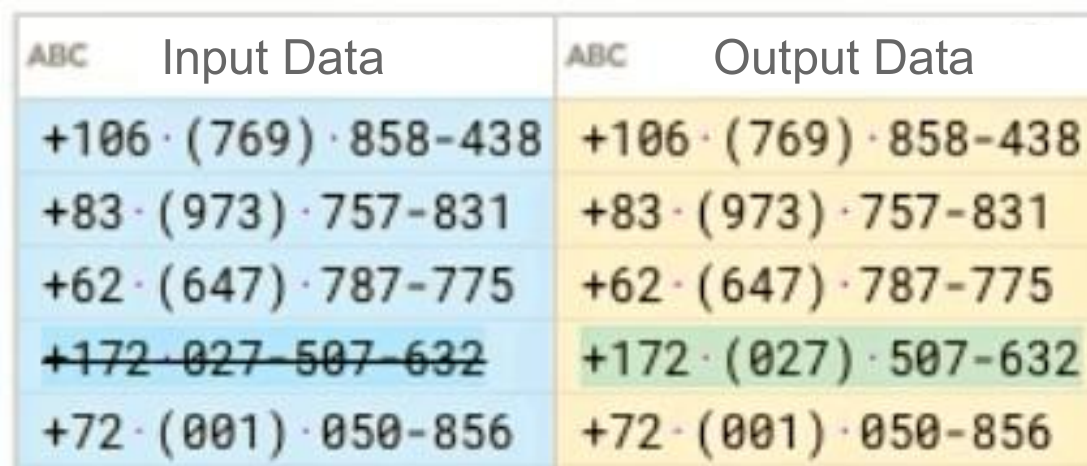}\vspace{0.6cm}
\caption{Preview Tab}\label{fig:preview}
\end{minipage}~
\begin{minipage}{0.55\textwidth}
\includegraphics[width=\linewidth]{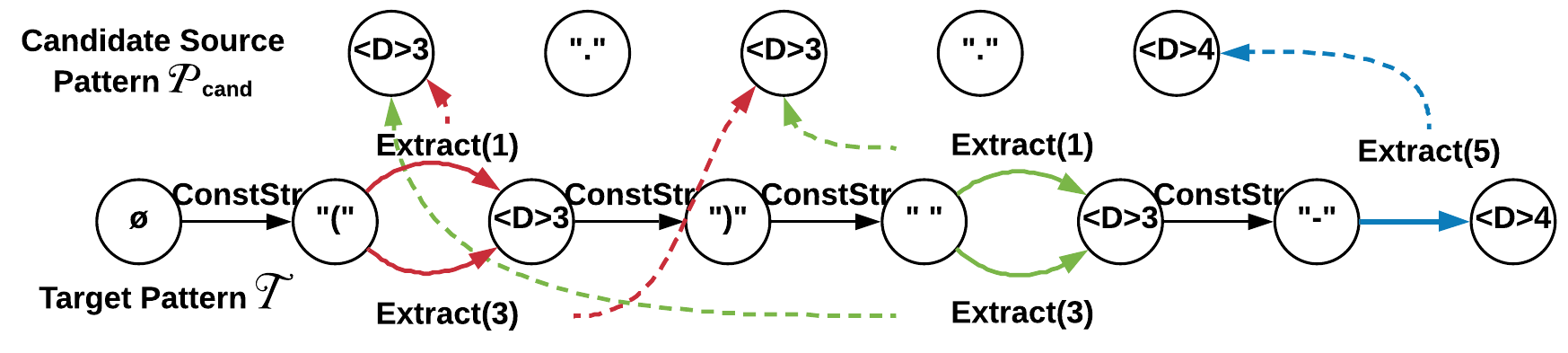}
\caption{Token alignment for the target pattern $\mathcal{T}$}
\label{fig:tokenmapping}
\end{minipage}~
\begin{minipage}{0.25\textwidth}
\centering
\includegraphics[width=\linewidth]{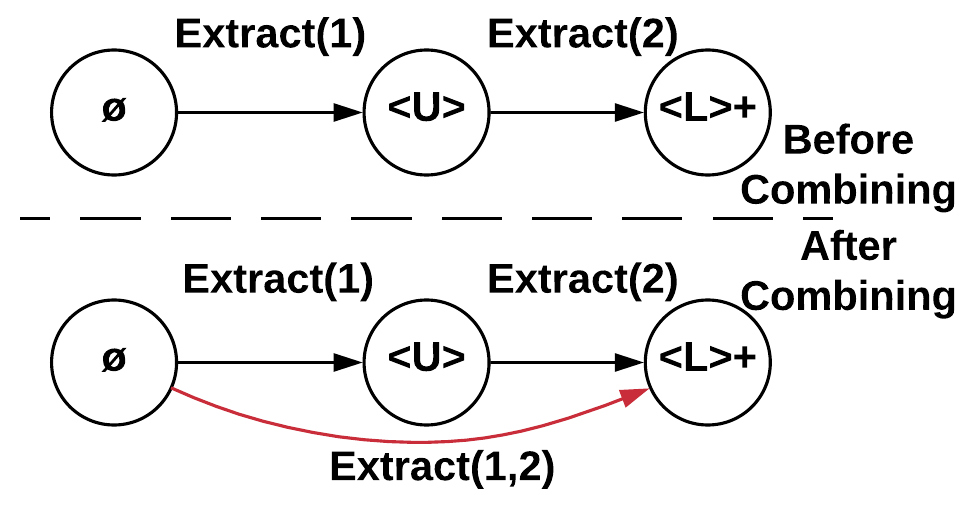}
\caption{Combine \textsf{Extracts}}\label{fig:combine}
\end{minipage}
\end{figure*}

Before synthesizing a transformation for a source pattern, we want to quickly check whether it can be a {\em candidate source pattern} (or source candidate), \ie, it is possible to find a transformation from this pattern into the target pattern, through $\texttt{validate}$. \textbf{If we can immediately disqualify some patterns, we do not need to go through more expensive data transformation synthesis process for them.}
There are a few reasons why some pattern in the hierarchy may not be qualified as a candidate source pattern:
\begin{enumerate}
\item The input data set may be ad hoc and a pattern in this data set can be a description of noise values. For example, a data set of phone numbers may contain ``N/A'' as a data record because the customer refused to reveal this information. In this case, it is meaningless to generate transformations.
\item We may be fundamentally not able to support some transformations (\eg, semantic transformations are not supported as in our case). Hence, we should filter out certain patterns which we think semantic transformation is unavoidable, because it is impossible to transform them into the desired pattern without the help from the user.
\item Some patterns are too general; it can be hard to determine how to transform these patterns into the target pattern. We can ignore them and create transformation plans for their children. For instance, if a pattern is ``$\langle AN \rangle$+,$\langle AN \rangle$+'', it is hard to tell if or how it could be transformed into the desired pattern of ``$\langle U \rangle \langle L \rangle+:\langle D \rangle+$''. By comparison, its child pattern  ``$\langle U \rangle\langle L \rangle+,\langle D \rangle+$'' seems to be a better fit as the candidate source. 
\end{enumerate}


\begin{algorithm}[t]
\SetKwData{pattern}{$\mathcal{P}$}
\SetKwData{target}{$\mathcal{T}$}
\SetKwData{map}{$\mathcal{G}$}
\SetKwData{unsolved}{$\mathcal{Q}_{unsolved}$}
\SetKwData{solved}{$\mathcal{Q}_{solved}$}
\SetKwData{root}{$\mathcal{P}_R$}
\SetKwData{return}{$\mathcal{L}$}

 \KwData{Pattern cluster hierarchy root \root, target pattern \target}
\KwResult{Synthesized program \return}
\unsolved, \solved $\leftarrow [\ ]$ \;
\return $\leftarrow \emptyset$\;
push \root to \unsolved\;
\While{$\unsolved \neq \emptyset$}{\label{algo:progSyn:1:start}
 $p \leftarrow$ pop \unsolved\;

\eIf{ $\texttt{validate}(p, \target) = \top$}{\label{algo:progSyn:1:validsource}
$\map \leftarrow \texttt{findTokenAlignment}(p,\target)$\;  \label{algo:progSyn:1}
push $\{p, \map\}$ to \solved \;
}{
push $p$.children to \unsolved\;\label{algo:progSyn:1:end}
}
}
\return $\leftarrow \texttt{createProgs}(\solved)$\;\label{algo:progSyn:2}
\textbf{Return} \return

 \caption{Synthesize \lang\ Program}
 \label{algo:progSyn}
 
\end{algorithm}

Any input data matching no candidate source pattern is left unchanged and flagged for additional review, which could involve replacing values with NULL or default values or manually overriding values.

Since the goal here is simply to quickly prune those patterns that are not good source patterns, the checking process should be able to find unqualified source patterns with {\em high precision} but not necessarily {\em high recall}. Here, we use a simple heuristic of {\em frequency count} that can effectively reject unqualified source patterns with high confidence: examining if there are sufficient base tokens of each class in the source pattern matching the base tokens in the target tokens. The intuition is that any source pattern with fewer base tokens than the target is unlikely to be transformable into the target pattern without external knowledge; base tokens usually carry semantic meanings and hence are likely to be hard to invent \emph{de novo}.

To apply frequency count on the source pattern $p_1$ and the target pattern $p_2$, $\texttt{validate}$ (denoted as $\mathcal{V}$) compares the \textit{token frequency} for every class of base tokens in $p_1$ and $p_2$. The token frequency $\mathcal{Q}$ of a token class $\langle \mathfrak{\widetilde{t}} \rangle$ in $p$ is defined as 
\begin{align}
\mathcal{Q}(\langle \mathfrak{\widetilde{t}} \rangle,p) = \sum_{i=1}^{n} \{t_{i}.\mathfrak{q} | t.name = \langle\mathfrak{\widetilde{t}}\rangle \}, p = \{t_1, \dots, t_n\} 
\end{align}
If a quantifier is not a natural number but ``+'', we treat it as 1 in computing $\mathcal{Q}$.

Suppose $\mathfrak{T}$ is the set of all token classes (in our case, $\mathfrak{T} = [\langle D \rangle, \langle L \rangle, \langle U \rangle, \langle A \rangle, \langle AN \rangle]$), $\mathcal{V}$ is then defined as
\begin{align}
\mathcal{V}(p_1, p_2) = \left\{ \begin{array}{rcl}
\text{true} & \mbox{if } \mathcal{Q}(\langle \mathfrak{\widetilde{t}}\rangle, p_1) \geq \mathcal{Q}(\langle \mathfrak{\widetilde{t}} \rangle, p_2), \forall \langle \mathfrak{\widetilde{t}} \rangle \in \mathfrak{T}\\
\text{false} & \mbox{otherwise}  
\end{array}\right.
\end{align}

\begin{example}
Suppose the target pattern $\mathcal{T}$ in Example~\ref{ex:missingRightBracket} is [`[', $\langle U\rangle$+, `-', $\langle D\rangle$+, `]'], we know 
\begin{align*}
\mathcal{Q}(\langle D\rangle, \mathcal{T}) = 
\mathcal{Q}(\langle U\rangle, \mathcal{T}) = 1
\end{align*} 
A pattern [`[', $\langle U\rangle$3, `-', $\langle D\rangle$5] derived from data record ``[CPT-00350'' will be identified as a source candidate by $\texttt{validate}$, because 
\begin{align*}
\mathcal{Q}(\langle D\rangle, p) &= 5 > \mathcal{Q}(\langle D\rangle, \mathcal{T})\ \wedge \\
\mathcal{Q}(\langle U\rangle, p) &= 3 > \mathcal{Q}(\langle U\rangle, \mathcal{T})
\end{align*}
Another pattern [`[', $\langle U\rangle$3, `-'] derived from data record ``[CPT-'' will be rejected because 
\begin{align*}
\mathcal{Q}(\langle D\rangle, p) = 0 < \mathcal{Q}(\langle D\rangle, \mathcal{T})
\end{align*}
\end{example}

\subsection{Token Alignment}\label{subsubsec:tokenAlign}

\newlength{\textfloatsepsave} 
\setlength{\textfloatsepsave}{\textfloatsep}
\setlength{\textfloatsep}{3pt}

\begin{algorithm}[t]
\SetKwData{src}{$\mathcal{P}_{cand}$}
\SetKwData{pattern}{$\mathcal{P}$}
\SetKwData{target}{$\mathcal{T}$}
\SetKwData{map}{$\mathcal{M}$}
\SetKwData{nodes}{$\widetilde{\eta}$}
\SetKwData{sNode}{$\eta^{s}$}
\SetKwData{tNode}{$\eta^{t}$}
\SetKwData{edges}{$\xi$}
\SetKwData{dag}{$\mathcal{G}$}

 \KwData{Target pattern \target\ = $\{t_1,\dots,t_m\}$, candidate source pattern \src\ = $\{t'_1,\dots,t'_n\}$, where $t_i$ and $t'_i$ denote base tokens}
 \KwResult{Directed acyclic graph \dag}
 \nodes $\leftarrow \{0, \dots, n\}$; \sNode $\leftarrow 0$; \tNode $\leftarrow n$; \edges $\leftarrow \{\}$\;
\For{$t_i \in \target$}{\label{algo:tokenMap:start}

\For{$t'_j \in \src$}{
\If{$\texttt{SyntacticallySimilar}(t_i, t'_j) = \top$ }{\label{algo:tokenMap:1:start}
$e \leftarrow \textsf{Extract}(t'_j)$\;
add $e$ to $\edges_{(i-1, i)}$\; \label{algo:tokenMap:1:end}
}
}
\If{$t_i.type = $ `literal'}{\label{algo:tokenMap:2:start}
$e \leftarrow \textsf{ConstStr}(t_i.name)$\;
add $e$ to $\edges_{(i-1, i)}$\;\label{algo:tokenMap:2:end}
}
}

\For{$i \in \{1, \dots, n - 1\}$}{\label{algo:tokenMap:3:start}
$\xi_{in} \leftarrow \{\forall e_p \in \xi_{(i - 1, i)}, e_p $ is an \textsf{Extract} operation\}\;\label{algo:tokenMap:3:1:start}
$\xi_{out} \leftarrow \{\forall e_q \in \xi_{(i, i + 1)}, e_q $ is an \textsf{Extract} operation\}\;
\For{$e_p \in \xi_{in}$}{
\For{$e_q \in \xi_{out}$}{
\If{$e_p.srcIdx + 1 = e_q.srcIdx$}{
$e \leftarrow$ \textsf{Extract}$(e_p.t_i, e_q.t_j)$\;
add $e$ to $\edges_{(i-1, i+1)}$\;\label{algo:tokenMap:3:end}
}
}
}

}

\dag $\leftarrow Dag(\nodes, \sNode, \tNode, \edges)$\;
\textbf{Return} \dag

 \caption{Token Alignment Algorithm}
 \label{algo:tokenMap}
 
\end{algorithm}

Once a source pattern is identified as a source candidate in Section~\ref{subsubsec:cand}, we need to synthesize an atomic transformation plan between this source pattern and the target pattern, which explains how to obtain the target pattern using the source pattern. To do this, we need to find the token matches for each token in the target pattern: discover all possible operations that yield a token. This process is called {\em token alignment}.




For each token in the target pattern, there might be multiple different token matches. Inspired by \cite{Gulwani2011}, we store the results of the token alignment in Directed Acyclic Graph (DAG) represented as a $DAG(\widetilde{\eta}, \eta^{s}, \eta^{t}, \xi)$ . $\widetilde{\eta}$ denotes all the nodes in DAG with $\eta^{s}$ as the source node and $\eta^{t}$ as the target node. Each node corresponds to a position in the pattern. $\xi$ are the edges between the nodes in $\widetilde{\eta}$ storing the source information, which yield the token(s) between the starting node and the ending node of the edge. Our proposed solution to token alignment in a DAG is presented in Algorithm~\ref{algo:tokenMap}. 

\minisection{Align Individual Tokens to Sources}
To discover sources, given the target pattern $\mathcal{T}$ and the candidate source pattern $\mathcal{P}_{cand}$,  we iterate through each token $t_{i}$ in $\mathcal{T}$ and compare $t_i$ with all the tokens in $\mathcal{P}_{cand}$. 

For any source token $t'_j$ in $\mathcal{P}_{cand}$ that is \textit{syntactically similar} (defined in Definition~\ref{defn:similar}) to the target token $t_{i}$ in $\mathcal{T}$, we create a token match between $t'_j$ and $t_{i}$ with an \textsf{Extract} operation on an edge from $t_{i-1}$ to $t_{i}$  (line~\ref{algo:tokenMap:start}-\ref{algo:tokenMap:2:end}).

\begin{defn}[Syntactically Similar]\label{defn:similar}
Two tokens $t_i$ and $t_j$ are syntactically similar if:
\begin{enumerate*}[label=\arabic*)]
\item they have the same class,
\item their quantifiers are identical natural numbers or one of them is `+' and the other is a natural number
\end{enumerate*}.
\end{defn}

When $t_{i}$ is a literal token, it is either a symbolic character or a constant value. To build such a token, we can simply use a \textsf{ConstStr} operation (line~\ref{algo:tokenMap:2:start}-\ref{algo:tokenMap:2:end}), instead of extracting it from the source pattern. This does not violate our previous assumption of not introducing any external knowledge during the transformation. 

\begin{example}
Let the candidate source pattern be [ $\langle D\rangle$3, `.', $\langle D\rangle$3, `.', $\langle D\rangle$4] and the target pattern be [`(', $\langle D\rangle$3, `)', ` ', $\langle D\rangle$3, `-', $\langle D\rangle$4]. Token alignment result for the source pattern $\mathcal{P}_{cand}$ and the target pattern $\mathcal{T}$, generated by Algorithm~\ref{algo:tokenMap} is shown in Figure~\ref{fig:tokenmapping}. In Figure~\ref{fig:tokenmapping}, a dashed line is a token match, indicating the token(s) in the source pattern that can formulate a token in the target pattern. A solid line embeds the actual operation in \lang\ rendering this token match. 
\end{example}

\minisection{Combine Sequential Extracts}
The \textsf{Extract} operator in our proposed language \lang\ is designed to extract one or more tokens sequentially from the source pattern. 
Line~\ref{algo:tokenMap:1:start}-\ref{algo:tokenMap:2:end} only discovers sources composed of an \textsf{Extract} operation generating an individual token. \textit{Sequential extracts} (\textsf{Extract} operations extracting multiple consecutive tokens from the source) are not discovered, and this token alignment solution is not complete. We need to find the \textit{sequential extracts}. 

Fortunately, discovering sequential extracts is not independent of the previous token alignment process; sequential extracts are combinations of individual extracts. With the alignment results $\xi$ generated previously, we iterate each state and combine every pair of \textsf{Extract}s on an incoming edge and an outgoing edge that extract two consecutive tokens in the source pattern (line~\ref{algo:tokenMap:3:start}-\ref{algo:tokenMap:3:end}). The \textsf{Extract}s are then added back to $\xi$. Figure~\ref{fig:combine} visualizes combining two sequential \textsf{Extracts}. The first half of the figure (titled ``Before Combining'') shows a transformation plan that generates a target pattern pattern $\langle U\rangle \langle D\rangle +$ with two operations--- \textsf{Extract}(1) and \textsf{Extract}(2). The second half of the figure (titled ``After Combining'') showcases merging the incoming edge and the outgoing edge (representing the previous two operations) and formulate a new operation (red arrow), \textsf{Extract}(1,2), as a combined operation of the two.

A benefit of discovering sequential extracts is it helps yield a ``simple'' program, as described in Section~\ref{subsec:rank}.



\minisection{Correctness}
Algorithm~\ref{algo:tokenMap} is \textit{sound} and \textit{complete}, which is proved in Appendix~\ref{app:correctness}.

\subsection{Program Synthesis using Token Alignment Result}\label{subsec:rank}
As we represent all token matches for a source pattern as a DAG (Algorithm~\ref{algo:tokenMap}), finding a transformation plan is to find a path from the initial state $0$ to the final state $l$, where $l$ is the length of the target pattern $\mathcal{T}$. 

The Breadth First Traversal algorithm can find all possible atomic transformation plans for this DAG. However, not all of these plans are equally likely to be correct and desired by the end user. The hope is to prioritize the correct plan. The Occam's razor principle suggests that the simplest explanation is usually correct. Here, we apply \textbf{Minimum Description Length} (MDL)~\cite{rissanen1978modeling}, a formalization of Occam's razor principle, to gauge the \textit{simplicity} of each possible program.

Suppose $\mathcal{M}$ is the set of models. In this case, it is the set of atomic transformation plans found given the source pattern $\mathcal{P}_{cand}$ and the target pattern $\mathcal{T}$. $\mathcal{E} = f_{1}f_{2}\dots f_{n} \in \mathcal{M}$ is an atomic transformation plan, where $f$ is a string expression. Inspired by~\cite{raman2001potter}, we define \textit{Description length} (DL) as follows:
\begin{align}
L(\mathcal{E}, \mathcal{T}) &= L(\mathcal{E}) + L(\mathcal{T}|\mathcal{E})
\end{align}

$L(\mathcal{E})$ is the \textit{model description length}, which is the length required to encode the model, and in this case, $\mathcal{E}$. Hence,
\begin{align}
L(\mathcal{E})  = |\mathcal{E}| \log m 
\end{align}
where $m$ is the number of distinct types of operations.

$L(\mathcal{T}|\mathcal{E})$ is the \textit{data description length}, which is the sum of the length required to encode $\mathcal{T}$ using the atomic transformation plan $\mathcal{E}$. Thus,
\begin{align}
L(\mathcal{T}|\mathcal{E}) = \sum_{f_i \in \mathcal{E}} \log{L(f_{i})} 
\end{align}
where $L(f_i)$ the length to encode the parameters for a single expression. For a \textsf{Extract(i)} or \textsf{Extract(i,j)} operation, $L(f) = \log |\mathcal{P}_{cand}|^{2}$ (recall \textsf{Extract(i)} is short for \textsf{Extract(i,i)}). For a \textsf{ConstStr}($\widetilde{s}$), $L(f) = \log c^{|\widetilde{s}|}$, where $c$ is the size of printable character set ($c = 95$).

With the concept of description length described, we define the minimum description length as 
\begin{align}
L_{min}(\mathcal{T},\mathcal{M}) = \min_{\mathcal{E} \in \mathcal{M}} \Big[L(\mathcal{E}) + L(\mathcal{T}|\mathcal{E})\Big]
\end{align}

In the end, we present the atomic transformation plan $\mathcal{E}$ with the minimum description length as the default transformation plan for the source pattern. Also, we list the other $k$ transformation plans with lowest description lengths.

\begin{example}
Suppose the source pattern is ``$\langle D \rangle2 / \langle D \rangle2 / \langle D \rangle4$  '', the target pattern $\mathcal{T}$ is ``$\langle D \rangle2/\langle D \rangle2$''. The description length of a transformation plan $\mathcal{E}_1=$  Concat(Extract(1,3)) is $L(\mathcal{E}_1, \mathcal{T}) = 1\log1 + 2\log3$. In comparison, the description length of another transformation plan $\mathcal{E}_2=$  Concat(Extract(1), ConstStr(`/'),Extract(3)) is $L(\mathcal{E}_2, \mathcal{T}) = 3\log2 + \log3^{2} + \log95 + \log3^{2} > L(\mathcal{E}_1, \mathcal{T})$. Hence, we prefer $\mathcal{E}_1$, a clearly simpler and better plan than $\mathcal{E}_2$.
\end{example}

\subsection{Limitations and Program Repair}

The target pattern $\mathcal{T}$ as the sole user input so far is  more ambiguous compared to input-output example pairs used in most other PBE systems. Also, we currently do not support ``semantic transformation''. We may face the issue of ``semantic ambiguity''---mismatching syntactically similar tokens with different semantic meanings. For example, if the goals is to transform a date of pattern ``DD/MM/YYYY'' into the pattern ''MM-DD-YYYY'' (our clustering algorithm works in this case). Our token alignment algorithm may create a match from ``DD'' in the first pattern to ``MM'' in the second pattern because they have the same pattern of $\langle D\rangle2$. The atomic transformation plan we \textit{initially} select for each source pattern can be a transformation that mistakenly converts ``DD/MM/YYYY'' into ``DD-MM-YYYY''. Although our algorithm described in Section~\ref{subsec:rank} often makes good guesses about the right matches, the system still infers an imperfect transformation about 50\% of the time (Appendix~\ref{app:userEffortDetails}). 

Fortunately, as our token alignment algorithm is complete and the program synthesis algorithm can discover all possible transformations and rank them in a smart way, the user can quickly find the correct transformation through {\em program repair}: replace the initial atomic transformation plan with another atomic transformation plans among the ones Section~\ref{subsec:rank} suggests for a given source pattern. 

To make the repair even simpler for the user, we deduplicate equivalent atomic transformation plans defined below before the repair phase.
\begin{defn}[Equivalent Plans]\label{defn:equivalence}
Two Transformation Plans are {\em equivalent} if, given the same source pattern, they always yield the same transformation result for any matching string.
\end{defn}
For instance, suppose the source pattern is [$\langle D \rangle 2$, `/', $\langle D \rangle 2$], if transformation plans $\mathcal{E}_1$ is [\textsf{Extract}(3), \textsf{Const}(`/'), \textsf{Extract}(1)] and transformation plans $\mathcal{E}_2$ is [\textsf{Extract}(3), \textsf{Extract}(2), \textsf{Extract}(1)], their final transformation result should be exactly the same and the only difference between  $\mathcal{E}_1$ and  $\mathcal{E}_2$ is the source of `/'. Presenting such equivalent transformations to the user will not be helpful but increase the user effort. Hence, we only pick the simplest plan in the same equivalence class and prune the rest. The methodology detecting the equivalent plans is elaborated in Appendix~\ref{app:equivalentplans}.

Overall, the repair process does not significantly increase the user effort. In those cases where the initial program is imperfect, 75\% of the time the user made just a single repair (Appendix~\ref{app:userEffortDetails}).

\section{Experiments}\label{sec:evaluation}
\definecolor{bblue}{HTML}{4F81BD}
\definecolor{rred}{HTML}{C0504D}
\definecolor{ggreen}{HTML}{9BBB59}
\definecolor{ppurple}{HTML}{9F4C7C}

\begin{figure*}[t]
\centering

\begin{subfigure}[t]{0.3\linewidth}
\centering
\begin{tikzpicture}

\begin{axis}[
width  = 0.95*\linewidth,
height = 3cm,
major x tick style = transparent,
ybar=2*\pgflinewidth,
bar width=6pt,
ymajorgrids = true,
ylabel = {Time(s)},
symbolic x coords={10(2), 100(4), 300(6)},
xtick = data,
ytick={0,100,...,500},
ymax = 550,
scaled y ticks = false,
enlarge x limits=0.25,
ymin=0,
legend cell align=left,
legend style={
legend columns=2,
at={(1.15,.9)},
anchor=south east,
column sep=1ex,
draw=none,
font=\small,
fill=none
}
]
\addplot[style={ggreen,fill=ggreen,mark=none,postaction={
        pattern=crosshatch dots
    }}]
coordinates {(10(2),82) (100(4),328) (300(6),389)};

\addplot[style={bblue,fill=bblue,mark=none,postaction={
        pattern=horizontal lines
    }}]
coordinates {(10(2),49) (100(4),168) (300(6),496)};

\addplot[style={rred,fill=rred,mark=none}]
coordinates {(10(2),38) (100(4),79) (300(6),85)};

\legend{\regexreplace, \baseline, \system}
\end{axis}
\end{tikzpicture}

\vspace{-0.2cm}
\caption{Overall completion time}
\label{fig:scalability}
\vspace{-0.2cm}
\end{subfigure}~\hspace{0.5cm}
\begin{subfigure}[t]{0.3\linewidth}
\centering
\begin{tikzpicture}
\begin{axis}[
width  = 0.95*\linewidth,
height = 3cm,
major x tick style = transparent,
ybar=2*\pgflinewidth,
bar width=6pt,
ymajorgrids = true,
ylabel = {Interaction \#},
symbolic x coords={10(2), 100(4), 300(6)},
xtick = data,
scaled y ticks = false,
enlarge x limits=0.25,
ymin=0,
legend cell align=left,
legend style={
legend columns=2,
at={(1.15,.9)},
anchor=south east,
column sep=1ex,
draw=none,
font=\small,
fill=none
}
]
\addplot[style={ggreen,fill=ggreen,mark=none,postaction={
        pattern=crosshatch dots
    }}]
coordinates {(10(2),1) (100(4),3) (300(6),3)};

\addplot[style={bblue,fill=bblue,mark=none,postaction={
        pattern=horizontal lines
    }}]
coordinates {(10(2),2) (100(4),3) (300(6),8)};

\addplot[style={rred,fill=rred,mark=none}]
coordinates {(10(2),2) (100(4),4) (300(6),5)};

\legend{\regexreplace, \baseline, \system}
\end{axis}
\end{tikzpicture}

\vspace{-0.2cm}
\caption{Rounds of interactions}
\label{fig:scalbilityIterationNumber}
\vspace{-0.2cm}
\end{subfigure}~\hspace{0.5cm}
\begin{subfigure}[t]{0.3\linewidth}
\centering
\begin{tikzpicture}

\begin{axis}[
		width  = 0.95\linewidth,
		height = 3cm,
		ymin = 0,
		ylabel = {time(s)},
		xlabel = {Interaction \#},
		xlabel style={
            at={(0.5,0.24)}
        },
		ytick={0,100,...,500},
		xmajorgrids = true,
		major y tick style = transparent,
		legend cell align=left,
legend style={
legend columns=2,
at={(1.15,.9)},
anchor=south east,
column sep=1ex,
draw=none,
font=\small,
fill=none
}
]
	
	\addplot [mark=star, mark options={fill=white, scale=0.7}, color=ggreen] table [y=trifacta, x=id, col sep=comma] {data/flashfill-scalability-timeline.csv};
	\addplot [mark=triangle*, mark options={fill=white, scale=0.7}, color=bblue] table [y=flashfill, x=id, col sep=comma] {data/flashfill-scalability-timeline.csv};
	\addplot [mark=*, mark options={fill=white, scale=0.7}, color=rred] table [y=clax, x=id, col sep=comma] {data/flashfill-scalability-timeline.csv};
	
	\legend{\regexreplace, \baseline, \system}
	\end{axis}
\end{tikzpicture}
\vspace{-0.6cm}
\caption{Interaction timestamps for 300(6)}
\label{fig:timeline}
\vspace{-0.2cm}
\end{subfigure}

\caption{Scalability of the system usability as data volume and heterogeneity increases (shorter bars are better)}
\end{figure*}
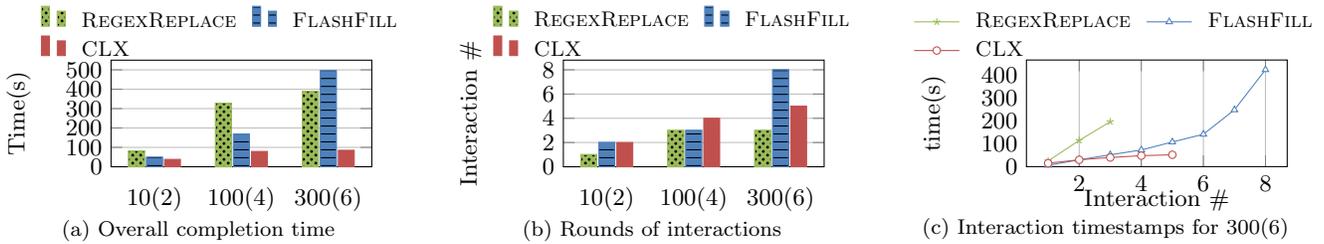
\definecolor{bblue}{HTML}{4F81BD}
\definecolor{rred}{HTML}{C0504D}
\definecolor{ggreen}{HTML}{9BBB59}
\definecolor{ppurple}{HTML}{9F4C7C}

\begin{figure*}[t]
\centering
\begin{minipage}{0.3\linewidth}
\begin{tikzpicture}

\begin{axis}[
width  = 0.95*\linewidth,
height = 3cm,
major x tick style = transparent,
ybar=2*\pgflinewidth,
bar width=6pt,
ymajorgrids = true,
ylabel = {Time(s)},
symbolic x coords={10(2), 100(4), 300(6)},
xtick = data,
ytick={0,100,...,500},
ymax = 550,
scaled y ticks = false,
enlarge x limits=0.25,
ymin=0,
legend cell align=left,
legend style={
legend columns=2,
at={(1.15,.9)},
anchor=south east,
column sep=1ex,
draw=none,
font=\small,
fill=none
}
]
\addplot[style={ggreen,fill=ggreen,mark=none,postaction={
        pattern=crosshatch dots
    }}]
coordinates {(10(2),33) (100(4),140) (300(6),266)};

\addplot[style={bblue,fill=bblue,mark=none,postaction={
        pattern=horizontal lines
    }}]
coordinates {(10(2),34) (100(4),142) (300(6),422)};

\addplot[style={rred,fill=rred,mark=none}]
coordinates {(10(2),36) (100(4),73) (300(6),81)};

\legend{\regexreplace, \baseline, \system}
\end{axis}
\end{tikzpicture}

\caption{Verification time (shorter bars are better)}
\label{fig:scalePause}
\end{minipage}~\hspace{0.5cm}
\begin{minipage}{0.3\linewidth}
\centering
\begin{tikzpicture}

\begin{axis}[
width  = 0.95*\linewidth,
height = 3cm,
major x tick style = transparent,
ybar=2*\pgflinewidth,
bar width=6pt,
ymajorgrids = true,
ylabel = {Correct rate},
symbolic x coords={task 1, task 2, task 3},
xtick = data,
ytick={0,0.2,...,1.0},
scaled y ticks = false,
enlarge x limits=0.25,
ymin=0,
legend cell align=left,
legend style={
legend columns=2,
at={(1.15,.9)},
anchor=south east,
column sep=1ex,
draw=none,
font=\small,
fill=none
}
]

\addplot[style={ggreen,fill=ggreen,mark=none,postaction={
        pattern=crosshatch dots
    }}]
coordinates {(task 1,0.89) (task 2,1) (task 3,0.89)};

\addplot[style={bblue,fill=bblue,mark=none,postaction={
        pattern=horizontal lines
    }}]
coordinates {(task 1,0.33) (task 2,0.67) (task 3,0.44)};

\addplot[style={rred,fill=rred,mark=none}]
coordinates {(task 1,1) (task 2,1) (task 3,0.89)};

\legend{\regexreplace , \baseline, \system}
\end{axis}
\end{tikzpicture}

\caption{User comprehension test (taller bars are better)}
\label{fig:comprehension}
\end{minipage}~\hspace{0.5cm}
\begin{minipage}{0.3\linewidth}
\centering
\begin{tikzpicture}
\begin{axis}[
width  = 0.95*\linewidth,
height = 3cm,
major x tick style = transparent,
ybar=2*\pgflinewidth,
bar width=6pt,
ymajorgrids = true,
ylabel = {Time(s)},
symbolic x coords={task 1, task 2, task 3},
xtick = data,
scaled y ticks = false,
enlarge x limits=0.25,
ymin=0,
legend cell align=left,
legend style={
legend columns=2,
at={(1.15,.9)},
anchor=south east,
column sep=1ex,
draw=none,
font=\small,
fill=none
}
]

\addplot[style={ggreen,fill=ggreen, postaction={
        pattern=crosshatch dots
    }, mark=none}]
coordinates {(task 1,97.67) (task 2,331.33) (task 3,311)};

\addplot[style={bblue,fill=bblue,postaction={
        pattern=horizontal lines
    },mark=none}]
coordinates {(task 1,115.33) (task 2,94.33) (task 3,191.33)};

\addplot[style={rred,fill=rred,mark=none}]
coordinates {(task 1,31.7) (task 2,132.33) (task 3,76)};

\legend{\regexreplace, \baseline, \system}
\end{axis}
\end{tikzpicture}
\caption{Completion time (shorter bars are better)}
\label{fig:userstudy}
\vspace{-0.2cm}
\end{minipage}
\end{figure*}

We make three broad sets of experimental claims. First, we show that as the input data becomes larger and messier, \system\ tends to be less work to use than \baseline\ because verification is less challenging (Section~\ref{subsec:scale}). Second, we show that \system\ programs are easier for users to understand than \baseline\ programs (Section~\ref{subsec:explanationEval}). Third, we show that \system's expressive power is similar to that of baseline systems, as is the required effort for non-verification portions of the PBE process (Section~\ref{subsec:expressiveness}). 



\subsection{Experimental Setup}
We implemented a prototype of \system\ and compared it against the state-of-the-art PBE system \baseline. For ease of explanation, in this section, we refer this prototype as ``\system''. Additionally, to make the experimental study more complete, we had a third baseline approach, a non-PBE feature offered by \trifactawrangler\footnote{\trifactawrangler\ is a commercial product of \wrangler\ launched by Trifacta Inc. The version we used is 3.2.1} allowing the user to perform string transformation through manually creating \textsf{Replace} operations with simple natural-language-like regexps (referred as \regexreplace). All experiments were performed on a 4-core Intel Core i7 2.8G CPU with 16GB RAM. Other related PBE systems, \foofah~\cite{Jin2017} and \tde~\cite{he2018transform}, target different workloads and also share the same verification problem we claim for PBE systems, and hence, are not considered as baselines.

\subsection{User Study on Verification Effort}\label{subsec:scale}

In this section, we conduct a user study on a real-world data set to show that
\begin{enumerate*}[(1)]
\item verification is a laborious and time-consuming step for users when using the classic PBE data transformation tool (\eg, \baseline) particularly on a large messy data set,
\item asking end users to hand-write regexp-based data transformation programs is challenging and inefficient, and 
\item the \system\ model we propose effectively saves the user effort in verification during data transformation and hence its interaction time does not grow fast as the size and the heterogeneity of the data increase.
\end{enumerate*}

\minisection{Test Data Set}
Finding public data sets with messy formats suitable for our experiments is very challenging. The first experiment uses a column of 331 messy phone numbers from the ``Times Square Food \& Beverage Locations'' data set~\cite{TimesSquareDataSet}. 

\minisection{Overview}
The task was to transform all phone numbers into the form ``$\langle D\rangle$3-$\langle D\rangle$3-$\langle D\rangle$4''. We created three test cases by randomly sampling the data set with the following data sizes and heterogeneity: ``10(2)'' has 10 data records and 2 patterns; ``100(4)'' has 100 data records and 4 patterns; ``300(6)'' has 300 data records and 6 patterns. 

We invited 9 students in Computer
Science with a basic understanding of regular expressions and not involved in our project. Before the study, we educated all participants on how to use the system. Then, each participant was asked to work on one test case on a system and we recorded their performance.

We looked into the user performances on three systems from various perspectives: {\em overall completion time}, {\em number of interactions}, and {\em verification time}. The \textit{overall completion time} gave us a quick idea of how much the cost of user effort was affected when the input data was increasingly large and heterogeneous in this data transformation task. The other two metrics allowed us to check the user effort in verification. While measuring completion time is straightforward, the other two metrics need to be clarified.

\informalsection{Number of interactions}
For \baseline, the number of interactions is essentially the number of examples the user provides. For \system\, we define the number of interactions as the number of times the user verifies (and repairs, if necessary) the inferred atomic transformation plans.  We also add one for the initial labeling interaction. For \regexreplace, the number of interactions is the number of \textsf{Replace} operations the user creates.

\informalsection{Verification Time}
All three systems follow different interaction paradigms. However, we can divide the interaction process of into two parts, \textit{verification} and {\em specification}: the user is either busy inputting (typing keyboards, selecting, etc.) or paused to verify the correctness of the transformed data or synthesized/hand-written regular expressions.

Measuring verification time is meaningful because we hypothesize that PBE data transformation systems become harder to use when data is large and messy not because the user has to provide a lot more input, but it becomes harder to verify the transformed data at the instance level.

\minisection{Results}
As shown in Figure~\ref{fig:scalability}, ``100(4)'' cost $1.1\times$ more time than ``10(2)'' on \system, and ``300(6)'' cost $1.2\times$ more time than ``10(2)'' on \system. As for \baseline, ``100(4)'' cost $2.4\times$ more time than ``10(2)'', and ``300(6)'' cost  $9.1\times$ more time than ``10(2)''. Thus, in this user study, the user effort required by \system\ grew slower than that of \baseline.  Also, \regexreplace\ cost significantly more user effort than \system\ but its cost grew not as quickly as \baseline.  This shows good evidence that 
\begin{enumerate*}[(1)]
\item manually writing data transformation script is cumbersome, 
\item the user interaction time grows very fast in \baseline\ when data size and heterogeneity increase, and 
\item the user interaction time in \system\ also grows, but not as fast.
\end{enumerate*}

Now, we dive deeper into understanding the causes for observation (2) and (3). Figure~\ref{fig:scalbilityIterationNumber} shows the number of interactions in all test cases on all systems. We see that all three systems required a similar number of interactions in the first two test cases. Although \baseline\ required 3 more interactions than \system\ in case ``300(6)'', this could hardly be the main reason why \baseline\ cost almost $5x$ more time than \system.

We take a close look at the three systems’ interactions in the case of ``300(6)'' and plot the timestamps of each interaction in Figure~\ref{fig:timeline}. The result shows that, in \baseline, as the user was getting close to achieving a perfect transformation, it took the user an increasingly longer amount of time to make an interaction with the system, whereas the interaction time intervals were relatively stable in \system\ and \regexreplace. Obviously, the user spent a longer time in each interaction NOT because an example became harder to type in (phone numbers have relatively similar lengths). We observed that, without any help from \baseline, the user had to eyeball the entire data set to identify the data records that were still not correctly transformed, and it became harder and harder to do so simply because there were fewer of them. Figure~\ref{fig:scalePause} presents the average verification time on all systems in each test case. ``100(4)'' cost $1.0\times$ more verification time than ``10(2)'' on \system, and ``300(6)'' cost $1.3\times$ more verification time than ``10(2)'' on \system. As for \baseline, ``100(4)'' cost $3.4\times$ more verification time than ``10(2)'', and ``300(6)'' cost  $11.4\times$ more verification time than ``10(2)''. The fact that the verification time on \baseline\ also grew significantly as the data became larger and messier supports our analysis and claim.

To summarize, this user study presents evidence that \baseline\ becomes much harder to use as the data becomes larger and messier mainly because verification is more challenging. In contrast, \system\ users generally are not affected by this issue.

\begin{table}[t]
\centering
\fontsize{8}{8}\selectfont
\begin{tabular}{c|cccc}
\toprule
Task ID & Size & AvgLen & MaxLen & DataType\\
\midrule
Task1&10&11.8&14&Human name\\
Task2&10&20.3&38&Address\\
Task3&100&16.6&18& Phone number\\
\bottomrule
\end{tabular}
\caption{Explainability test cases details}
\label{tab:comprehensibilityData}
\end{table}

\subsection{User Study on Explainability}\label{subsec:explanationEval}

Through a new user study with the same 9 participants on three tasks, we demonstrate that
\begin{enumerate*}[(1)]
\item \baseline\ users lack understanding about the inferred transformation logic, and hence, have inadequate insights on how the logic will work, and show that
\item the simple program generated by \system\ improves the user's understanding of the inferred transformation logic.
\end{enumerate*}

Additionally, we also compared the overall completion time of three systems. 

\minisection{Test Set}
Since it was impractical to give a user too many data pattern transformation tasks to solve, we had to limit this user study to just a few tasks. To make a fair user study, we chose tasks with various data types that cost relatively same user effort on all three systems. 
From the benchmark test set we will introduce in Section~\ref{subsec:expressiveness}, we randomly chose 3 test cases that each is supposed to require same user effort on both \system\ and \baseline: Example 11 from FlashFill (task 1), Example 3 from PredProg (task 2) and ``phone-10-long'' from SyGus (task 3). Statistics (number of rows, average/max/min string length of the raw data) about the three data sets are shown in Table~\ref{tab:comprehensibilityData}. 

\minisection{Overview}
We designed 3 multiple choice questions for every task examining how well the user understood the transformation regardless of the system he/she interacted with. All the questions were formulated as ``Given the input string as $x$, what is the expected output''. All questions are shown in Appendix~\ref{app:questions}.

During the user study, we asked every participant to participate all three tasks, each on a different system (completion time was measured). Upon completion, each participant was asked to answer all questions based on the transformation results or the synthetic programs generated by the system. 

\minisection{Explainability Results}
The correct rates for all 3 tasks using all systems are presented in Figure~\ref{fig:comprehension}. The result shows that the participants were able to answer these questions almost perfectly using \system, but struggled to get even half correct using \baseline. \regexreplace\ also achieved a success rate similar to \system, but required higher user effort and expertise. 

The result suggests that \baseline\ users have insufficient understanding about the inferred transformation logic and \system\ improves the users' understanding in all tasks, which provides evidence that verification in \system\ can be easier.

\minisection{Overall Completion Time}
The average completion time for each task using all three systems is presented in Figure~\ref{fig:userstudy}. Compared to \baseline, the participants using \system\ spent 30\% less time on average: $\sim70\%$ less time on task 1 and  $\sim60\%$ less time on task 3, but $\sim40\%$ more time on task 2. Task 1 and task 3 have similar heterogeneity but task 3 (100 records) is bigger than task 1 (10 records). The participants using \baseline\ typically spent much more time on understanding the data formats at the beginning and verifying the transformation result in solving task 3. This provides more evidence that \system\ saves the verification effort. Task 2 is small (10 data records) but heterogeneous. Both \baseline\ and \system\ made imperfect transformation logic synthesis, and the participants had to make several corrections or repairs. We believe \system\ lost in this case simply because the data set is too small, and as a result, \system\ was not able to exploit its advantage in saving user effort on large-scale data set.
The study also gives evidence that \system\ is sometimes effective in saving user verification effort in small-scale data transformation tasks.

\begin{table}[t]
\centering
\setlength{\tabcolsep}{2pt}
\fontsize{7}{7}\selectfont
\begin{tabularx}{\linewidth}{lllllX}
\toprule
Sources & \# tests & AvgSize&AvgLen & MaxLen & DataType\\
\midrule
SyGus~\cite{SyGuS}&27&63.3&11.8&63& car model ids, human name, phone number, university name and address  \\
FlashFill~\cite{Gulwani2011}&10&10.3&15.8&57& log entry, phone number, human name, date, name and position, file directory, url, product name  \\
BlinkFill~\cite{Singh2016}&4&10.8&14.9&37& city name and country, human name, product id, address \\
PredProg~\cite{singh2015predicting}&3&10.0&12.7&38& human name, address \\
Prose~\cite{prose}&3&39.3&10.2&44& country and number, email, human name and affiliation\\
\midrule
Overall&47&43.6&13.0&63&\\
\bottomrule
\end{tabularx}
\caption{Benchmark test cases details}
\label{tab:benchmarkSize}
\end{table}

\subsection{Expressivity and Efficiency Tests}\label{subsec:expressiveness}
In a simulation test using a large benchmark test set, we demonstrate that
\begin{enumerate*}[(1)]
\item the expressive power of \system\ is comparable to the other two baseline systems \baseline\ and \regexreplace, and
\item \system\ is also pretty efficient in costing user interaction effort.
\end{enumerate*}

\minisection{Test Set}
We created a benchmark of 47 data pattern transformation test cases using a mixture of public string transformation test sets and example tasks from related research publications (will be released upon the acceptance of the paper). The information about the number of test cases from each source, average raw input data size (number of rows), average/max data instance length, and data types of these test cases are shown in Table~\ref{tab:benchmarkSize}.  A detailed description of the benchmark test set is shown in Appendix~\ref{app:benchmarks}.

\minisection{Overview}
We evaluated \system\ against 47 benchmark tests. As conducting an actual user study on all 47 benchmarks is not feasible, we simulated a user following the ``lazy approach'' used by Gulwani \ea~\cite{harris2011spreadsheet}: a simulated user selected a target pattern or multiple target patterns and then repaired the atomic transformation plan for each source pattern if the system proposed answer was imperfect. 

Also, we tested the other two systems against the same benchmark test suite. As with \system, we simulated a user on \baseline; this user provided the first positive example on the first data record in a non-standard pattern, and then iteratively provided positive examples for the data record on which the synthetic string transformation program failed. On \regexreplace, the simulated user specified a \textsf{Replace} operation with two regular expressions indicating the matching string pattern and the transformed pattern, and iteratively specified new parameterized \textsf{Replace} operations for the next ill-formatted data record until all data were in the correct format.

\minisection{Evaluation Metrics}
In experiments, we measured how much user effort all three systems required. Because systems follow different interaction models, a direct comparison of the user effort is impossible. We quantify the user effort by {\em Step}, which is defined differently as follows

\begin{itemize}
    \item For \system, the total Steps is the sum of the number of correct patterns the user chooses (Selection) and the number of repairs for the source patterns whose default atomic transformation plans are incorrect (Repair). In the end, we also check if the system has synthesized a ``perfect'' program: a program that successfully transforms all data.
    \item For \baseline, the total Steps is  the sum of the number of input examples to provide and the number of data records that the system fails to transform.
    \item For \regexreplace, each specified \textsf{Replace} operation is counted as 2 Steps as the user needs to type two regular expressions for each \textsf{Replace}, which is about twice the effort of giving an example in \baseline. 
\end{itemize}
In each test, for any system, if not all data records were correctly transformed, we added the number of data records that the system fails to transform correctly to its total \textit{Step} value as a punishment.
In this way, we had a coarse estimation of the user effort in all three systems on the 47 benchmarks.

\begin{table}[t]
\centering
\fontsize{8}{8}\selectfont
\begin{tabular}{c|ccc}
\toprule
Baselines& \system\ Wins&Tie& \system\ Loses\\
\midrule
vs.\ \baseline& 17 (36\%)&17 (36\%) &13 (28\%)\\
vs.\ \regexreplace& 33 (70\%)& 12 (26\%)& 2 (4\%)\\
\bottomrule
\end{tabular}
\caption{User effort simulation comparison. }
\label{tab:winorlose}
\end{table}

\minisection{Expressivity Results}
\system\ could synthesize right transformations for 42/47 ($\sim 90\%$) test cases, whereas \baseline\ reached 45/47 ($\sim 96\%$). This suggests that the expressive power of \system\ is comparable to that of \baseline. 

There were five test cases where \system\ failed to yield a perfect transformation. Only one of the failures was due to the expressiveness of the language itself, the others could be fixed if there were more representative examples in the raw data. ``Example 13'' in FlashFill requires the inference of advanced conditionals (\textsf{Contains} keyword ``picture'') that \lang\ cannot currently express, but adding support for these conditionals in \lang\ is straightforward. The failures in the remaining four test cases were mainly caused by the lack of the target pattern examples in the data set. For example, one of the test cases we failed is a name transformation task, where there is a last name ``McMillan'' to extract. However, all data in the target pattern contained last names comprising one uppercase letter followed by multiple lowercase letters and hence our system did not realize ``McMillan'' needed to be extracted. We think if the input data is large and representative enough, we should be able to successfully capture all desired data patterns.

\regexreplace\ allows the user to specify any regular expression replace operations, hence it was able to correctly transform all the input data existed in the test set, because the user could directly write operations replacing the exact string of an individual data record into its desired form. However, similar to \lang, \regexreplace\ is also limited by the expressive power of regular expressions and cannot support advanced conditionals. As such, it covered 46/47 ($\sim98\%$) test cases.

\minisection{User Effort Results}
As the \textit{Step} metric is a potentially noisy measure of user effort, it is more reasonable to check whether \system\ costs more or less effort than other baselines, rather than to compare absolute \textit{Step} numbers. The aggregated result is shown in Table~\ref{tab:winorlose}. It suggests \system\ often requires less or at least equal user effort than both PBE systems. Compared to \regexreplace, \system\ almost always costs less or equal user effort. A detailed discussion about the user effort on \system\ and comparison with other systems is in Appendix~\ref{app:userEffortDetails}.

\section{Related Work}\label{sec:related}
\minisection{Data Transformation}
\baseline\ (now a feature in Excel) is an influential work for syntactic transformation by Gulwani~\cite{Gulwani2011}. It designed an expressive string transformation language and proposed the algorithm based on version space algebra to discover a program in the designed language. It was recently integrated to PROSE SDK released by Microsoft. A more recent PBE project, \tde~\cite{he2018transform}, also targets string transformation. Similar to \baseline, \tde\ requires the user to verify at the instance level and the generated program is unexplainable to the user. Other related PBE data cleaning projects include \cite{Singh2016,Jin2017}. 

Another thread of seminal research including \cite{raman2001potter}, \wrangler~\cite{Kandel2011} and \trifacta\ created by Hellerstein \etal\ follow a different interaction paradigm called ``predictive interaction''. They proposed an inference-enhanced visual platform supporting many different data wrangling and profiling tasks. Based on the user selection of columns, rows or text, the system intelligently suggests possible data transformation operations, such as  \textsf{Split}, \textsf{Fold}, or pattern-based extraction operations.


\minisection{Pattern Profiling}
In our project, we focus on clustering ad hoc string data based on structures and derive the structure information. The \learnpads~\cite{Fisher2008} project is somewhat related. It presents a learning algorithm using statistics over symbols and tokenized data chunks to discover pattern structure. \learnpads\ assumes that all data entries follow a repeating high-level pattern structure. However, this assumption may not hold for some of the workload elements. In contrast, we create a bottom-up pattern discovery algorithm that does not make this assumption. Plus, the output of \learnpads\ (\ie, PADS program~\cite{Fisher2005}) is hard for a human to read, whereas our pattern cluster hierarchy is simpler to understand. Most recently, \datamaran\cite{Gao2017NavigatingTD} has proposed methodologies for discovering structure information in a data set whose record boundaries are unknown, but for the same reasons as \learnpads, \datamaran\ is not suitable for our problem.

\minisection{Program Synthesis} Program synthesis has garnered wide interest in domains where the end users might not have good programming skills or programs are hard to maintain or reuse including data science and database systems. Researchers have built various program synthesis applications to generate SQL queries~\cite{wang2017synthesizing, qian2012sample, li2014constructing}, regular expressions~\cite{blackwell2001swyn,li2008regular}, data cleaning programs~\cite{Gulwani2011,wu2015iterative}, and more.



Researchers have proposed various techniques for program synthesis. \cite{gulwani2011synthesis, jha2010oracle} proposed a constraint-based program
synthesis technique using logic solvers. However, constraint-based techniques are mainly applicable in the context where finding a satisfying solution is challenging, but we prefer a high-quality program rather than a satisfying program. Version space algebra is another important technique that is applied by \cite{mitchell1982generalization,lau2003programming,Gulwani2011,lau2000version}. \cite{devlin2017robustfill} recently focuses on using deep learning for program synthesis. Most of these projects rely on user inputs to reduce the search space until a quality program can be discovered; they share the hope that there is one simple solution matching most, if not all, user-provided example pairs. In our case, transformation plans for different heterogeneous patterns can be quite distinct. Thus, applying the version space algebra technique is difficult.

\section{Conclusion and Future Work}\label{sec:conclusion}
Data transformation is a difficult human-intensive task. PBE is a leading approach of using computational inference to reduce human burden in data transformation. However, we observe that standard PBE for data transformation is still difficult to use due to its laborious and unreliable verification process.

We proposed a new data transformation paradigm \system\ to alleviate the above issue. In \system, we build data patterns to help the user quickly identify both well-formatted and ill-formatted data which immediately saves the verification time. \system\ also infers regexp replace operations as the desired transformation, which many users are familiar with and boosts their confidence in verification.

We presented an instantiation of \system\ with a focus on data pattern transformation including
\begin{enumerate*}[(1)]
\item a pattern profiling algorithm that hierarchically clusters both the raw input data and the transformed data based on data patterns,
\item a DSL, \lang, that can express many data pattern transformation tasks and can be interpreted as a set of simple regular expression replace operations,
\item algorithms inferring a correct \lang\ program.
\end{enumerate*}

We presented two user studies. In a user study on data sets of various sizes, when the data size grew by a factor of 30, the user verification time required by \system\ grew by $1.3\times$ whereas that required by \baseline\ grew by $11.4\times$. The comprehensibility user study shows the \system\ users achieved a success rate about twice that of the \baseline\ users. The results provide good evidence that \system\ greatly alleviates the verification issue.

Although building a highly-expressive data pattern transformation tool is not the central goal of this paper, we are happy to see that the expressive power and user effort efficiency of our initial design of \system\ is comparable to those of \baseline\ in a simulation study on a large test set in another test.

\system\ is a data transformation paradigm that can be used not only for data pattern transformation but other data transformation or transformation tasks too. For example, given a set of heterogeneous spreadsheet tables storing the same information from different organizations, \system\ can be used to synthesize programs converting all tables into the same standard format. Building such an instantiation of \system\ will be our future work.

\bibliographystyle{plain}
\bibliography{paper}

\begin{appendix}
\section{Correctness of Token Alignment Algorithm}\label{app:correctness}
\begin{theorem} [Soundness]
If the token alignment algorithm (Algorithm~\ref{algo:tokenMap}) successfully discovers a token correspondence, it can be transformed into a \lang\ program.
\end{theorem}
\begin{proof}
Recall that an atomic transformation plan for a pair of source pattern and target pattern is a concatenation of \textsf{Extract} or \textsf{ConstStr} operations that sequentially generates each token in the target pattern. Every token correspondence discovered in Algorithm~\ref{algo:tokenMap} corresponds to either a \textsf{ConstStr} operation or a \textsf{Extract}, both of which will generate one or several tokens in the target pattern. Hence, a token correspondence is can be possibly admitted into an atomic transformation plan, which will end up becoming part of a \lang\ program. The soundness is true.
\end{proof}
\begin{theorem} [Completeness]
If there exists a \lang\ program, the token alignment algorithm (Algorithm~\ref{algo:tokenMap}) will for sure discover the corresponding token correspondence matching the program.
\end{theorem}
\begin{proof}
Given the definition of the \lang\ and candidate source patterns, the completeness is true only when the token alignment algorithm can discover all possible parameterized \textsf{Extract} and/or \textsf{ConstStr} operations which combined will generate all tokens for the target pattern. In Algorithm~\ref{algo:tokenMap}, line~\ref{algo:tokenMap:1:start}-\ref{algo:tokenMap:1:end} is certain to discover any \textsf{Extract} operation that extracts a single token in the source pattern and produces a single token in a target pattern; line~\ref{algo:tokenMap:2:start}-\ref{algo:tokenMap:2:end} guarantees to discover any \textsf{ConstStr} operation that yields a single constant token in a target pattern. Given the design of our pattern profiling, an \textsf{Extract} of a single source token can not produce multiple target tokens, because such multiple target tokens, if exist, must have the same token class, and should be merged as one token whose quantifier is the sum of the all these tokens. Similarly, the reverse is also true. What remains to prove is whether Algorithm~\ref{algo:tokenMap} is guaranteed to generate an \textsf{Extract} of multiple tokens, \ie, \textsf{Extract}$(p,q) (p < q)$, in the source pattern that produces multiple tokens in the target pattern. In Algorithm~\ref{algo:tokenMap}, line~\ref{algo:tokenMap:1:start}-\ref{algo:tokenMap:1:end} is guaranteed to discover \textsf{Extract}$(p)$, \textsf{Extract}$(p+1)$, $\dots$, \textsf{Extract}$(q)$. With these \textsf{Extracts}, when performing line~\ref{algo:tokenMap:3:1:start}-\ref{algo:tokenMap:3:end} when $i = p+1$ in Algorithm~\ref{algo:tokenMap}, it will discover the incoming edge representing \textsf{Extract}$(p)$ and the output edge representing \textsf{Extract}$(p+1)$ and combine them, generating \textsf{Extract}$(p, p+1)$. When $i = p + 2$, it will discover the incoming edge representing \textsf{Extract}$(p, p+1)$ and the outgoing edge representing \textsf{Extract}$(p+2)$ and combine them, generating \textsf{Extract}$(p, p+2)$. If we repeat this process, we will definitely find \textsf{Extract}$(p,q)$ in the end. Therefore, the solution is complete.
\end{proof}


\section{Equivalent Plans Detection}\label{app:equivalentplans}
To ``deduplicate'' a list of candidate transformation plans $P_1, P_2, \dots, P_q$, we first pairwise compare the first plan $P_1$ with the rest of the plans, and if $P_1$ has any equivalent plans, we only keep the one that is the simplest (see Section~\ref{subsec:rank}), and remove the rest. We then repeat the previous process until we remove the duplicates for all plans in the list. The computational complexity of the above process is $\mathcal{O}(q^2)$, where $q$ is the number of plans in the list. In practice, $q$ is usually a small number. Hence, the above deduplication process is reasonably inexpensive.

Checking whether a candidate transformation plan $P_1$ is equivalent to another candidate transformation plan $P_2$ is performed through the following procedures:
\begin{enumerate}
\itemsep0em 
\item Split each \textsf{Extract}$(m,n)$ operation in both plans into \textsf{Extract}$(m)$,\textsf{Extract}$(m+1)$, $\dots$, \textsf{Extract}$(n)$. 
\item Assuming $P_1 = \{op_1^{1}, op_2^{1}, \dots, op_n^{1}\}$ and $P_2 = \{op_1^{2}, op_2^{2}, \dots, op_m^{2}\}$. If $m \neq n$, we stop checking and return $\mathbf{False}$. Otherwise, from left to right, we compare operations of two plans one by one. For example, we first compare $op_1^{1}$ with $op_1^{2}$, then  $op_2^{1}$ with $op_2^{2}$, and so on. The check continues when 
\begin{enumerate}
\itemsep0em 
\item $op_k^{1}$ is exactly the same as $op_k^{2}$, or
\item $op_k^{1}$ is not same as $op_k^{2}$. However, one of them is an \textsf{Extract} operation and the other is a \textsf{ConstStr} operation, and the first operation extracts a constant string whose content is exactly same as the content of the second operation.
\end{enumerate}.
\item We stop and return $\mathbf{True}$ if we reach the end of both plans.
\end{enumerate}

The computational complexity of above pairwise comparison is clearly linear to the length of the plan, and is therefore inexpensive.

\section{Questions used in Program Explanation Experiment}\label{app:questions}

\begin{enumerate}
\item For task 1, if the input string is ``Barack Obama'', what is the output?
  \begin{choices}
    \choice Obama
    \choice Barack, Obama
    \choice Obama, Barack
    \choice None of the above
  \end{choices}
\item For task 1, if the input string is ``Barack Hussein Obama'', what is the output?
  \begin{choices}
    \choice Obama, Barack Hussein
    \choice Obama, Barack
    \choice Obama, Hussein
    \choice None of the above
  \end{choices}

\item For task 1, if the input string is ``Obama, Barack Hussein'', what is the output?
  \begin{choices}
    \choice Obama, Barack Hussein
    \choice Obama, Barack
    \choice Obama, Hussein
    \choice None of the above
  \end{choices}
  
\item For task 2, if the input is “155 Main St, San Diego, CA 92173”, what is the output
  \begin{choices}
    \choice San
    \choice San Diego
    \choice St, San
    \choice None of the above
  \end{choices}

\item For task 2, if the input string is “14820 NE 36th Street, Redmond, WA 98052”, what is the output?
  \begin{choices}
    \choice Redmond
    \choice WA
    \choice Street, Redmond
    \choice None of the above
  \end{choices}

\item For task 2, if the input is “12 South Michigan Ave, Chicago”, what is the output?
  \begin{choices}
    \choice South Michigan
    \choice Chicago
    \choice Ave, Chicago
    \choice None of the above
  \end{choices}
  
\item For task 3, if the input string is “+1 (844) 332-282”, what is the output?
  \begin{choices}
    \choice +1 (844) 282-332
    \choice +1 (844) 332-282
    \choice +1 (844)332-282 
    \choice None of the above
  \end{choices}

\item For task 3, if the input string is “844.332.282”, what is the output?
  \begin{choices}
    \choice +844 (332)-282
    \choice +844 (332) 332-282
    \choice +1 (844) 332-282
    \choice None of the above
  \end{choices}

\item For task 3, if the input string is “+1 (844) 332-282 ext57”, what is the output?
  \begin{choices}
    \choice +1 (844) 322-282
    \choice +1 (844) 322-282 ext57
    \choice +1 (844) 282-282 ext57 
    \choice None of the above
  \end{choices} 
\end{enumerate}

\section{Creation of Benchmark Test Set}\label{app:benchmarks}

Among the 47 test cases we collected, 27 are from SyGus (Syntax-guided
Synthesis Competition), which is a program synthesis contest held every year. In 2017, SyGus revealed 108 string transformation tasks in its Programming by Examples Track: 27 unique scenarios and 4 tasks of different sizes for each scenario. We collected the task with the longest data set in each scenario and formulated the pattern normalization benchmarks of 27 tasks. We collected 10 tasks from FlashFill~\cite{Gulwani2011}. There are 14 in their paper. Four tests (Example 4, 5, 6, 14) require a loop structure in the transformation program which is not supported in \lang\ and we filter them out. Additionally, we collected 4 tasks from \blinkfill~\cite{Singh2016}, 3 tasks from PredProg~\cite{singh2015predicting}, 3 tasks from Microsoft PROSE SDK~\cite{prose}.

For test scenarios with very little data, we asked a Computer Science student not involved with this project to synthesize more data. Thus, we have sufficient data for evaluation later. Also, the current 
\system\ prototype system requires at least one data record in the target pattern. For any benchmark task, if the input data set violated this assumption, we randomly converted a few data records into the desired format and used these transformed data records and the original input data to formulate the new input data set for the benchmark task.  The heterogeneity of our benchmark tests comes from the input data and their diverse pattern representations in the pattern language described previously in the paper. 
\definecolor{bblue}{HTML}{4F81BD}
\definecolor{rred}{HTML}{C0504D}
\definecolor{ggreen}{HTML}{9BBB59}
\definecolor{ppurple}{HTML}{9F4C7C}

\begin{figure}[!t]
\centering
\begin{subfigure}{\linewidth}
\centering
\begin{tikzpicture}
\begin{axis}[
ybar,
xmin = 1,
xmax = 47,
bar width=1pt,
width  = \linewidth,
height = 2.8cm,
ymin=0,
ymax = 6,
xtick={1,5,...,47},
minor xtick={1,2,...,47},
tickwidth=0,
ytick={1,2,3,4,5},
ymajorgrids=true,
xtick align=outside,
ytick align=outside,
     yticklabel=\pgfmathparse{\tick}\pgfmathprintnumber{\pgfmathresult}$\times$,
]
\addplot [color=rred, fill=rred] table [x=x,  y=Flashfill, col sep=comma] {data/efficiency-times.csv};

\end{axis}
\end{tikzpicture}

\vspace{-0.2cm}
\caption{\system\ vs.\ \baseline}
\label{fig:speedupflashfill}
\end{subfigure}
\begin{subfigure}{\linewidth}
\centering
\begin{tikzpicture}
\begin{axis}[
ybar,
xmin = 1,
xmax = 47,
bar width=1pt,
width  = \linewidth,
height = 2.8cm,
ymin=0,
ymax = 6,
xtick={1,5,...,47},
minor xtick={1,2,...,47},
tickwidth=0,
ytick={1,2,3,4,5},
ymajorgrids=true,
xtick align=outside,
ytick align=outside,
yticklabel=\pgfmathparse{\tick}\pgfmathprintnumber{\pgfmathresult}$\times$,
]
\addplot [mark=none, color=ggreen, fill=ggreen] table [x=x,  y=Trifacta, col sep=comma] {data/efficiency-times.csv};

\end{axis}
\end{tikzpicture}
\caption{\system\ vs.\ \regexreplace}
\label{fig:speeduptrifacta}
\end{subfigure}

\caption{Speedup: \# of \textit{Step}s ratio for 47 test cases}
\label{fig:efficiency}

\end{figure}
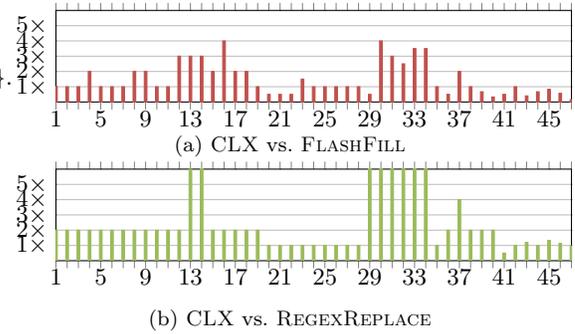

\section{User Effort Detailed Analysis}\label{app:userEffortDetails}

Figure~\ref{fig:efficiency} shows the overall \textit{speedup} of \system\ over the other two baselines. The y-axis is
the value of \textit{speedup}: number of \textit{Steps} cost in \system\ over the other baseline system. The x-axis
denotes the benchmark id. A summarized result showing the percentage of test cases that cost more or less effort on \system\ compared to the other baselines is shown in Table~\ref{tab:winorlose}. It suggests \system\ often requires less or at least equal user effort than both PBE systems. Compared to \regexreplace, \system\ almost always costs less or equal user effort.

Figure~\ref{fig:breakdown} is a breakdown of the user effort required by \system. The y-axis is the number of \textit{Steps}; the x-axis denotes the percentage of test cases that costs less than or equal to the number of \textit{Steps} indicated by y-axis. For around 79\% of the test cases, \system\ is able to infer a perfect data normalization program within two steps. Also, the user needs to select only one target pattern in the initial step for about 79\% of the test cases, which proves that our pattern profiling technique is usually effective in grouping data under the same pattern. 

Additionally, we observe that the user needs to make no adjustment for the suggested transformations in about 50\% of the test cases and $\leq 1$ adjustment in about 85\% of the test cases. This shows that Occam's razor principle we follow and the algorithm we design is effective in prioritizing the correct transformations and producing quality results. 

Note that, in test case ``popl-13.ecr'' from PROSE, \system\ consumed tremendous user effort because the data are a combination of human names, organization names and country names. All these names do not share a distinctive syntax for us to identify. For example, if the user wants to extract both ``INRIA'' and ``Univ. of California'', the user might have to select both ``$\langle U\rangle$+'' and ``$\langle U\rangle\langle L\rangle+.\ \langle L\rangle+\ \langle U\rangle\langle L\rangle+$'' as the target patterns, and adjust more later. This increases the user selection and adjustment effort. 


\definecolor{bblue}{HTML}{4F81BD}
\definecolor{rred}{HTML}{C0504D}
\definecolor{ggreen}{HTML}{9BBB59}
\definecolor{ppurple}{HTML}{9F4C7C}

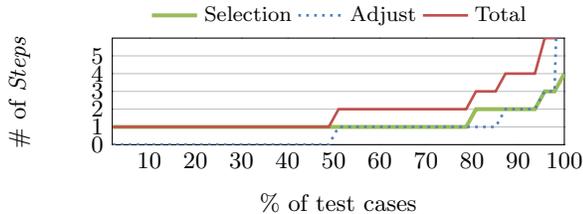
\begin{figure}[!t]
\centering
\begin{tikzpicture}
\begin{axis}[
	width  = 0.9\linewidth,
  ymin=0,
  ymax=6,
  xmin=2,
  xmax=100,
  height=3cm,
    xlabel={\% of test cases},
    ylabel={\# of \textit{Steps}},
    xtick={0,10,...,100},
    ytick={0,1,...,5},
    tickwidth=0,
    ymajorgrids=true,
    xtick align=outside,
    ytick align=outside,
    legend style={anchor=north east, cells={anchor=west}, font=\scriptsize,yshift=0.6cm, xshift=-0.2cm, legend columns=3,draw=none, fill=none,font=\small},
    ]

\addplot [mark=none, color=ggreen,line width=1.5pt] table [x=x, x expr=\thisrowno{0}/47*100, y=selection, col sep=comma]  {data/breakdown.csv};
\addplot [mark=none, dotted, color=bblue,line width=1pt] table [x=x, x expr=\thisrowno{0}/47*100, y=reconfigure, col sep=comma] {data/breakdown.csv};
\addplot [mark=none,color=rred,line width=1pt] table [x=x, x expr=\thisrowno{0}/47*100, y=unifacta, col sep=comma] {data/breakdown.csv};

\legend{Selection, Adjust, Total}
\end{axis}
\end{tikzpicture}
\vspace{-.2cm}
\caption{Percentage of test cases costing $\leq$ Y \textit{Steps} in different \system\ interaction phases}\label{fig:breakdown}
\end{figure}

However, this problem can be easily solved by suggesting a single operation of \textsf{Extract} between two commas for all source patterns once we identify the ``comma'' is a ``StructProphecy''~\footnote{A pattern with ``StructProphency'' tokens is a pattern of $k$ fields separated by $k-1$ struct tokens; in ``popl-13.ecr'' from PROSE, all source data are three name fields separated by two commas, and we want to extract the field in the middle} using the methodology proposed by~\cite{Fisher2008}. In this case, the user effort is substantially reduced.

\end{appendix}

\end{document}